\documentclass{article}
\usepackage{graphicx} 
\usepackage[parfill]{parskip}

\usepackage{graphicx} 
\usepackage{float}
\usepackage{amsthm}
\usepackage{amsmath, amssymb, mathtools}
\usepackage{extarrows}
\usepackage{tikz-cd}
\usepackage{hyperref}
\usepackage[backend=biber, sorting=none, style=alphabetic]{biblatex}
\usepackage{geometry}
\usepackage{tikz}
\usetikzlibrary{calc}

\DeclareMathOperator{\hor}{hor}
\DeclareMathOperator{\pr}{pr}

\DeclareMathOperator{\Spin}{Spin}
\DeclareMathOperator{\GL}{GL}
\DeclareMathOperator{\SL}{SL}
\DeclareMathOperator{\SO}{SO}
\DeclareMathOperator{\SU}{SU}

\DeclareMathOperator{\tr}{tr}

\DeclareMathOperator{\Fr}{Fr}
\DeclareMathOperator{\Sym}{Sym}
\DeclareMathOperator{\pt}{pt}

\newcommand{\spin}{\mathfrak{spin}}

\newcommand{\Z}{\mathbb{Z}}
\newcommand{\vol}{\mathrm{vol}}

\newcommand{\into}{\xhookrightarrow{}}

\newcommand{\Dirac}{D\!\!\!\!/\,}

\DeclareMathOperator{\Der}{Der}

\DeclareMathOperator{\End}{End}

\DeclareMathOperator{\ad}{ad}

\DeclareMathOperator{\vspan}{span}
\newcommand{\diff}{\mathrm{d}}
\DeclareMathOperator{\id}{id}

\DeclareMathOperator{\Hom}{Hom}

\newcommand{\pdiff}[2]{\frac{\partial #1}{\partial #2}}

\newcommand{\C}{\mathbb{C}}
\newcommand{\Ber}{{\mathcal B}er}

\renewcommand{\i}{\mathrm{i}}

\theoremstyle{definition}
\newtheorem*{remark}{Remark}
\newtheorem{definition}{Definition}
\numberwithin{definition}{section}
\newtheorem{example}{Example}

\theoremstyle{plain}
\newtheorem{principle}{Principle}
\newtheorem{theorem}{Theorem}

\newtheorem{lemma}{Lemma}

\newtheorem{corollary}{Corollary}

\newcommand{\imgSuperspaceTorsion}{
\begin{tikzpicture}
\def\nlines{8}

\coordinate (a) at (0,0,0);
\coordinate (b) at (0,1,5);
\coordinate (c) at (4,0,0);
\coordinate (d) at (4.5,-1,4);
\coordinate (e) at (4.5,2,4);

\coordinate (x) at ($(a)!3/2!(b)$);
\coordinate (y) at ($(a)!3/2!(c)$);
\coordinate (z) at (0,3,0);
\coordinate (z2) at (0,-2,0);

\draw[blue, -stealth] (a) -- (z) node[anchor=west]{$\mathcal B$};
\draw[blue] (a) -- (z2);

\draw[blue, thick, dotted] (b) -- (e);
\draw[blue, thick, dotted] (c) -- (d);

\foreach \t in {1,...,\nlines}
{
    \draw[gray, ultra thin, -stealth] ($(a)!\t / \nlines!(c)$) -- ($(b)! \t / \nlines!(e)$);
    \draw[gray, ultra thin, -stealth] ($(a)!\t / \nlines!(b)$) -- ($(c)! \t / \nlines!(d)$);
}

\draw[black, thick] (a) -- node[anchor=south]{$X$} (b);
\draw[black, thick, -stealth] (b) -- node[anchor=north]{$Y$} (d);
\draw[black, thick] (a) -- node[anchor=south]{$Y$} (c);
\draw[black, thick, -stealth] (c) -- node[anchor=south]{$X$} (e);
\draw[blue, thick, -stealth] (d) -- node[anchor=west]{$\pr_\mathcal B[X,Y]$} (e);
\filldraw[black] (0,0) circle (2pt);
\end{tikzpicture}
}

\title{Spinorial Superspaces and Super Yang-Mills Theories}
\author{Johannes Moerland}
\date{\today}

\begin{document}

\maketitle

\begin{abstract}
    In physics literature about supersymmetry, many authors refer to \emph{super Minkowski spaces}. These spaces are affine supermanifolds with certain distinguished spin structures. In these notes, we make the notion of such spin structures precise and generalise the setup to curved supermanifolds. This leads to the more general notion of \emph{spinorial superspaces}. By working in a suitable geometric and coordinate-free setting, many explicit coordinate computations appearing in physics literature can be replaced by more conceptual methods. As an application of the rather general framework of spinorial superspaces, we formulate $\mathcal N = 1$ super Yang-Mills theories on curved superspaces of spacetime dimensions $d=3$ and $d=4$ and show how to reduce the theory to field theories defined on an underlying ordinary spacetime manifold.
\end{abstract}
\tableofcontents

\section{Introduction}Modern differential geometry allows for an inherently coordinate-free approach to physical field theories. The theory of smooth manifolds and principal bundles is a suitable framework to define gauge theories such as Yang-Mills theories. In this language, symmetries such as gauge invariance are often manifest in the theories of interest. The upshot of this approach is that lengthy and cumbersome coordinate computations can often be replaced by more conceptual algebraic computations, commutative diagrams and exact sequences. This often allows for a rather pleasant geometric picture.

To extend this geometric approach to super field theories, a generalisation of the category of smooth manifolds is desirable. In sight of the anticommutativity properties of fermionic particles, the `even' (ordinary) directions should be complemented by `odd' (nilpotent) directions which are described by anti-commuting (Grassmannian) coordinate functions. By making this notion precise in an algebro-geometric picture, one arrives at the category of \emph{supermanifolds}. In this work, we make use of the language initially suggested by Leites~\cite{leites1980introduction} as well as functorial geometric extension as introduced by Kessler~\cite{kessler2018supergeometry}. Further comprehensive and introductory texts are presented in~\cite{deligne1999quantum, varadarajan2004supersymmetry}.

While the theory of supermanifolds is well-established, it remains a highly non-trivial task to cast
super field theories in this language. There are several reasons for this: Firstly, when physicists
speak about superspaces, more structure than the mere notion of a supermanifold is usually implied.
In particular, a certain spin structure plays an important role. Secondly, the computational steps
giving validity to super field theories (for example, coordinate independence) are often carried out in
long computations involving a plethora of indices, and are hence difficult to keep track of. Thirdly,
we would like to mention that many constructions on Minkowski space do not generalise canonically,
if at all) to curved supermanifolds.

In this document, we make the notion of the aforementioned extra structure on superspaces precise, leading us to the notion of spinorial superspaces. We discuss emergent operators and maps such as a canonical connection and partial pseudo-Riemannian metrics. We show how to reduce field theories defined on superspaces to ones defined on ordinary spacetimes. Furthermore, we explicitly construct spinorial superspaces in terms of ordinary smooth manifolds. Finally, we showcase super Yang-Mills theories in $d=3$ and $d=4$ spacetime dimensions as theories defined on spinorial superspaces.

This document is structured as follows: in section~\ref{section:superspace}, we give a geometric and mathematically precise definition of superspaces. 

In section~\ref{sec:spin-structures}, we define certain distinguished spin structures which globalise the notion of the super translation algebra, leading us to the definition of spinorial superspaces. We analyse the implications of such spin structures by providing a description of the space of connections and listing invariant morphisms on the tangent bundle.

Depending on the spacetime dimension and signature, the spinorial part of the complexified tangent sheaf of the superspace decomposes into chiral and antichiral components. This is discussed in section~\ref{section:chiral-decomposition}.

In section~\ref{section:restriction}, we show that any embedded purely even manifold, thought of as ordinary spacetime, is canonically equipped with a plethora of structures such as a pseudo-Riemannian metric and a spinor bundle.

We then construct spinorial superspaces by means of ordinary Riemannian spin manifolds in section~\ref{sec:split-superspaces}. 

As an application of the theory, we formulate super Yang-Mills theories in spacetime dimensions $3$ and $4$ in Lorentzian signature in section~\ref{section:sym}.
\section{Superspaces}\label{section:superspace}
Physical field theories are often expressed explicitly in terms of coordinates. While this approach allows for simple comparison of theoretical and experimental data, it lacks a certain intuition. What gives the theories geometric significance is that they are invariant under change of coordinates and/or gauge transformations.\par
Differential geometry and the theory of principal $G$-bundles admit an inherently invariant description of field theories on manifolds. This bears two major advantages: Firstly, upon phrasing the theories accordingly in this framework, one dodges the infamous `battle of indices' which is rather prone to miscalculations and sign errors. Secondly, the invariant formulation in terms of manifolds and principal bundles generalises physical field theories -- which are often defined explicitly in terms of coordinates on Minkowski space -- straight-forwardly to topologically non-trivial spaces.\par
In sight of the required anticommutativity properties of fermionic particles, it is thus sensible to extend the theory of manifolds to that of $\mathbb Z_2$-graded manifolds, better known as \emph{supermanifolds}. In what follows (up to principle~\ref{princ:smfd-fop}) we summarise the essential notions in the theory of super differential geometry (or in short, \emph{supergeometry}) following Kessler~\cite{kessler2018supergeometry}.
\begin{definition}
A \emph{supermanifold of dimension $m|n$} is a locally ringed space $M=(\underline M,\mathcal O _M)$, where each $x\in\underline M$ has an open neighbourhood $U$ such that the restriction of the sheaf of sections $\mathcal O _M$ to $U$ is isomorphic to
\begin{equation}
\mathcal O _M(U) \cong C^\infty(\hat U)\otimes\wedge(\mathbb R^n)^\vee
\end{equation}
as local $\mathbb Z _2$-graded ring. Here, $\hat U \subset \mathbb R^m$ is an open subset.
\end{definition}
Note that this definition characterises the supermanifold $M$ by its sheaf of sections over the topological space $\underline M$. Roughly speaking, this is necessary as the odd coordinate functions are nilpotent, and therefore, they cannot correspond to directions supported on some topological space.
\begin{definition}
Let $M,N$ be supermanifolds. A \emph{morphism of supermanifolds} $f:M\to N$ is a tuple $f=(\underline f, f^*)$, where $\underline f:\underline M \to \underline N$ is a continuous map, and the map of sheaves $f^* : \mathcal O _N \to \underline f ^*\mathcal O _M$ preserves parity and the maximal ideals of the stalks. Together with this notion of morphisms, supermanifolds constitute a category, denoted  $\mathbf{sMfd}$.
\end{definition}
Note that the definition implies that the morphisms are purely even. To allow for a mixture of odd and even directions on the supermanifold, we use the language of Ref.~\cite{kessler2018supergeometry}, where the author introduces families of supermanifolds:
\begin{definition}
A \emph{family of supermanifolds} is a surjective submersion of supermanifolds $M\to A$. We say that $M$ has \emph{relative dimension} $\dim(M)-\dim(A)$.
\end{definition}
Note that since we defined supermanifolds as locally ringed spaces, they technically do not possess isolated geometrical points: In fact, we can only look at the stalks, which capture infinitesimal neighbourhoods of points. However, upon considering families of supermanifolds $M\to A$, we can consider \emph{$A$-points} of $M$, which are sections $A\to M$ of the submersion. By viewing all supermanifolds over $A$ in this way, we can use the odd directions of $A$ to mix even and odd components of $M$ and $N$ in morphisms $f:M\to N$. In other words, we use the auxiliary base space $A$ to parametrise the supermanifold $M$.
\begin{definition}
We shall denote by $\mathbf{sMfd}_A$ the category of supermanifolds over $A$. The homset is given by morphisms of supermanifolds $M\to N$ that preserve the base $A$.
\end{definition}
We can replace the base by pulling back the submersion $\sigma : M\to A$ along a surjective submersion of supermanifolds $\varphi: A \to \tilde A$, giving rise to a covariant functor
\begin{equation}
\mathbf{sMfd}_A \to \mathbf{sMfd}_{\tilde A},\quad (M,\sigma, A)\mapsto (M,\sigma\circ\varphi,\tilde A).
\end{equation}
\begin{principle}\label{princ:smfd-fop}
Henceforth, we shall implicitly view each supermanifold over a sufficiently large base $A$. In particular, we work exclusively in the category $\mathbf{sMfd}_A$. Whenever we mention the dimension of $M$, we shall refer to the relative dimension of $M$ over $A$.
\end{principle}
As shown in Ref.~\cite{kessler2018supergeometry}, this geometrical implementation of the functor-of-points allows for notions of tangent bundles, flows of vector fields, super Lie groups, principal fibre bundles, principal connections and many more. In particular, we shall use the principal bundle of frames and replacements of its structure group\footnote{The misnomer \emph{reductions of the structure group} is usually used; but the supporting group homomorphisms need not be injective in general, such that we resort to the use of this terminology.} to express super gauge theories geometrically.\par
We shall now give a enlightening example as to why the auxiliary base is essential to make sense of some geometrical constructions.
Recall that the \emph{flow of a vector field} $X$ is given by the family of diffeomorphisms $\phi^X_t: M\to M$ such that
\begin{equation}
X(f)\circ \phi_t^X = \frac{\diff}{\diff t} f\circ\phi_t^X.
\end{equation}
\begin{example}
Let $M$ be a supermanifold, and denote by $U\subset M$ a chart domain.
In what follows, we consider flows of local coordinate vector fields. To that end, let $(x^a,\xi^\alpha)$ denote coordinate functions on $U$. Let us denote by $\mathcal O_U$ the restriction of $\mathcal O _M$ to $U$, and similarly, write $\mathcal T _{U}$ for the restricted tangent sheaf. Note that
\begin{equation}
\mathcal T _{U} = \vspan_{\mathcal O _{U}} \left\{\pdiff{}{x^a},\pdiff{}{\xi^\alpha}\right\} = \vspan_{\mathcal O _{U}} \{\pdiff{}{X_A}\},
\end{equation}
where
\begin{equation}X_A = \begin{cases}{x^a} & \textrm{if } A=a,\\
{\xi^\alpha} & \textrm{if } A=\alpha.\end{cases}
\end{equation}
An explicit computation shows that in these coordinates, the flow $\phi^{\pdiff{}{X^B}}_t\eqqcolon \phi^B_t$ is given by
\begin{equation}
\phi^B_t: U\to U,\quad (X^1,\hdots,X^{B-1},X^B,X^{B+1},\hdots)\mapsto (X^1,\hdots,X^{B-1},X^B + t,X^{B+1},\hdots),
\end{equation}
that is, for short, $\phi^B_t : X^A \mapsto X^A + t\delta^{AB}$. In particular, whenever $X^B$ is an odd coordinate, $t$ needs to be an odd parameter. Thus, instead of taking $\mathbb R$ as ground ring, we require some graded ring $R$ with non-trivial odd component. The submersion $U\to A$  in the category $\mathbf{sMfd}_A$ allows us to take $R=\mathcal O _A$, giving the ground ring geometric significance.
\end{example}

We now adapt Kapranov's definition of a superspace (see Ref.~\cite{kapranov2021supergeometry}) in this framework, and specialise it as follows:
\begin{definition}
We say that a supermanifold $M$ of dimension $m|n$ is a \emph{superspace} if it is equipped with a decomposition of the tangent sheaf
\begin{equation}
\mathcal T _M = \mathcal B \oplus \mathcal F,
\end{equation}
such that $\mathcal F$ has full odd rank $0|n$ and is maximally non-integrable in the sense that the Frobenius map
\begin{equation}
\wedge^2_{\mathcal O _M} \mathcal F \to \mathcal B,\quad
(X\wedge Y) \mapsto [X,Y] \mod \mathcal F
\end{equation}
is surjective.
\end{definition}
Recall that by the Serre-Swan theorem for supermanifolds (see, e.g., Ref.~\cite{kessler2018supergeometry}), locally free $\mathcal{O}_M$-modules correspond uniquely to (locally free) vector bundles over $M$, such that the section functor
\begin{equation}
\mathbf{lfVB}_M \xrightarrow{\Gamma(M,-)}\mathbf{lfMod}_{\mathcal O _M}, \quad E \mapsto \Gamma(M,E)\eqqcolon \mathcal E
\end{equation}
defines an equivalence of categories. In this spirit, we shall write $TM\to M$ for the vector bundle that corresponds to the tangent sheaf, and likewise, $F$ and $B$ for the regular distributions that correspond to $\mathcal F$ and $\mathcal B$, that is,
\begin{equation}
\mathcal B = \Gamma(M,B),\quad \mathcal F = \Gamma(M,F).
\end{equation}
We will freely choose between the sheaf language and the vector bundle language, depending on which description suits us best.

Let us now elaborate on the second point of this definition. To that end, recall that once we specify a distribution $\mathcal F$, we have the $\mathcal O _M$-bilinear \emph{Frobenius map}
\begin{equation}
\wedge^2_{\mathcal O_M}\mathcal F \to \mathcal T _M / \mathcal F,\quad X\wedge Y \mapsto [X,Y]\mod\mathcal F.
\end{equation}
This pairing is precisely the obstruction to $\mathcal F$ being \emph{integrable}, i.e., the obstruction to $\mathcal F$ arising as a foliation of $M$ by subsupermanifolds $N_\alpha$ such that the restricted tangent sheaves $\left.\mathcal T _M\right\vert_{N_\alpha}$ arise as
\begin{equation}\left.\mathcal T_M\right\vert_{N_\alpha}\cong \left.\mathcal O _M\right\vert_{N_\alpha}\otimes \mathcal T_{N_\alpha}.
\end{equation}

Note that in our case, we additionally have orthogonal projections $\pr_\mathcal F: \mathcal T_M \to \mathcal F$ and $\pr_\mathcal B : \mathcal T _M \to \mathcal B$, such that we can identify the quotient $\mathcal T _M /\mathcal F$ with the orthogonal complement $\mathcal B$. This allows us to make the Frobenius obstruction explicit:
\begin{equation}
f_F:\wedge^2_{\mathcal O _M}\mathcal F \to \mathcal B,\quad X\wedge Y\mapsto \pr_{\mathcal B}([X,Y]).
\end{equation}

To give physical context, elements of $\mathcal B\subset\mathcal T_M$ will be thought of as bosonic fields, and elements of $\mathcal F\subset\mathcal T _M$ will be thought of as fermionic (spinorial) fields. We shall call $\mathcal B$ the \emph{bosonic distribution} and $\mathcal F$ the \emph{spinorial distribution}. The Frobenius map is the obstruction to $\mathcal F$ being integrable, and can be thought of as a \emph{current}:
\begin{figure}[H]
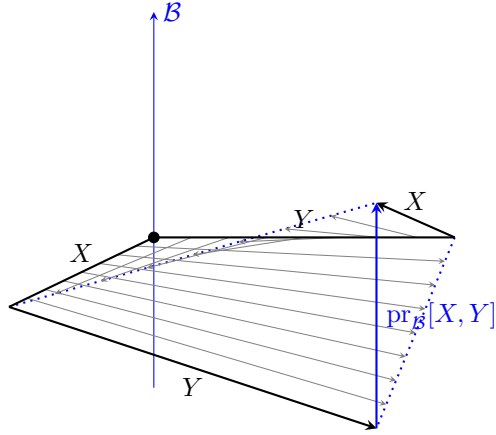

    \centering
    \imgSuperspaceTorsion
    \caption{The commutator of two fermionic fields attains contributions along $\mathcal B$.}
    \label{fig:superspace-frobenius}
\end{figure}

Note that the decomposition of the tangent sheaf equips the de Rham complex with a cohomological bigrading
\begin{equation}
    \Omega^r(M) = \bigoplus_{p+q=r}\Omega^{p|q}(M),\quad
    \Omega^{p|q}(M) = \wedge^p_{\mathcal O _M}\mathcal B^\vee\,\hat\otimes\,\wedge^q_{\mathcal O _M}\mathcal F^\vee,
\end{equation}
where `$\hat\otimes$' denotes the graded tensor product.\par
The prescription of the two complementary subbundles of $\mathcal T _M$ amounts to a reduction of the structure group of the bundle of frames $\GL(m|0)\times\GL(0|n)\to\GL(m|n)$, and we shall denote the reduced frame bundle by
\begin{equation}
Q^B\times_M Q^F \to \Fr(TM) \to M.
\end{equation}
Here, we consider $\GL(m|n)$ as the super Lie group of even automorphisms of $\mathbb R^{m|n}$, viewed as a trivial family of supermanifolds over the auxiliary base $A$. The reduction decomposes the soldering form as follows: Let $p=(p_B,p_F)$ be a local section of $(Q^B\times_M Q^F)|_U \to U \subset M$. Such a section is by definition a local basis of $B\oplus F$, such that it amounts to the respective isomorphisms
\begin{equation}
p_B: \mathbb R^{m|0} \times U \to B|_U,\quad p_F: \mathbb R^{0|n} \times U \to F|_U.
\end{equation}
Hence, the reduced soldering form on $Q^B\times_M Q^F$ reads
\begin{align}\label{eq:soldering-decomposition}
\theta_{(p_B,p_F)} &= (p_B,p_F)^{-1}\circ\diff\pi = p_B^{-1}\circ\pr_B^*\diff\pi \oplus p_\mathcal F^{-1}\circ\pr_F^*\diff\pi \\&= \theta_{p_ B}^B \oplus \theta_{p_F}^F \in \Omega^1(Q^B \times_M Q^F, \mathbb R^{m|0}\oplus\mathbb R^{0|n})^{\GL(m)\times\GL(n)}_{\hor}.
\end{align}

The following lemma provides a feasible way to compute the Frobenius obstruction explicitly:
\begin{lemma}
Let $\varphi$ be a $\GL(m)\times\GL(n)$ connection on the superspace, let $\Theta^{\varphi}$ denote its torsion, and let $X,Y\in\mathcal T_{Q^B \times_M Q^F}$ be the horizontal lifts of $\tilde X,\tilde Y\in\mathcal F \subset \mathcal T _M$ with respect to $\varphi$. It holds
\begin{equation}
\pr_{\mathbb R^{m|0}} \left(\Theta^{\varphi}(X,Y)\right) = \iota_{[X,Y]}\theta^B,
\end{equation}
irrespective of the choice of $\varphi$.
\end{lemma}
\begin{proof}
Let $X,Y$ be horizontal lifts of elements of $\mathcal F$. It holds
\begin{equation}
\pr_{\mathbb R^{m|0}} \Theta^{\varphi}(X,Y) = \iota_X\iota_Y\diff^{\varphi}\theta^B = \iota_X L_Y \theta^B = \iota_{[X,Y]}\theta^B,
\end{equation}
where we used that $\theta^\mathcal B$ vanishes on $\GL(m|0) \times\GL(0|n)$-invariant lifts of elements of $\mathcal F$ (see \eqref{eq:soldering-decomposition}). In particular, the result does not depend on the choice of connection. 
\end{proof}

\section{Spin Structures}\label{sec:spin-structures}
In what follows, we develop the notion of \emph{spinorial superspaces}. The latter intertwine the spinor algebra with the properties of the distribution $\mathcal F \subset\mathcal T _M$, thus allowing for a pleasant geometric interpretation of fermionic particles.\par
Throughout, we shall denote by $(V,q)$ the quadratic vector space that infinitesimally models the ordinary spacetime manifold, and by $S$ a real spinorial representation of $\Spin(V,q)$.
Any irreducible representation $S_0$ has a \emph{type} $\mathbb F\in\{\mathbb R, \mathbb C, \mathbb H\}$, such that the compact group of unitary elements $K^\mathbb F = \{x\in\mathbb F : |x|^2=1\}$ is (isomorphic to) the $\Spin(V,q)$-invariant automorphisms of $S_0$. In what follows, we shall assume that $(V,q)$ is Minkowskian. The following theorems provide existence and uniqueness results of the $\Gamma$ maps.
\begin{theorem}[\cite{deligne1999quantum}]
    Let $S_0$ be an irreducible real spinorial representation of $(V,q)$, where $q$ is Minkowskian. There exists, up to rescaling, a unique real symmetric morphism of representations $\Gamma: S_0\otimes S_0 \to V$. It holds
    \begin{itemize}
        \item $\Gamma$ is invariant under the action of $K^\mathbb F\subset \mathbb F$ on $S_0$.
        \item Let $v\in V$ such that $q(v,v)> 0$. The bilinear form $q(v,\Gamma(-,-))$ is non-degenerate.
\end{itemize}\label{thm:gamma-real}
\end{theorem}
\begin{theorem}[\cite{deligne1999quantum}]
    Let $S$ be a real representation of $\Spin(V,q)$ (not necessarily irreducible), and $\Gamma:S\otimes S \to V$ a symmetric map of representations of $\Spin(V,q)$. It holds:
    \begin{itemize}
        \item $S=\bigoplus_i S^{(i)}$ for real irreducible spinorial representations $S^{(i)}$, and $\Gamma = \sum_i \Gamma^{(i)}$, where the $\Gamma^{(i)}$ are as in theorem~\ref{thm:gamma-real}.
        \item If each of the $\Gamma^{(i)}$ is non-zero, there exists a unique $\tilde\Gamma : S^\vee\otimes S^\vee \to V$ such that the duals $\gamma:V\to\Hom(S,S^\vee)$, $\tilde\gamma: V\to \Hom(S^\vee, S)$ (after employing the musical isomorphism due to $q$ on $V$) obey the Clifford relations
        \begin{align}
            \tilde \gamma (v) \circ\gamma (w) 
            +\tilde \gamma (w) \circ\gamma (v)
            &= q(v,w)\id_{S^\vee},
            \\
            \gamma (v) \circ \tilde\gamma (w) +
            \gamma (w) \circ \tilde\gamma (v)&= q(v,w)\id_{S}.
        \end{align}
    \end{itemize}
    \label{thm:gamma-real-unique}
\end{theorem}
The relevance of these maps lies in the \emph{super translation algebra}:
\begin{definition}\label{def:super-translation}
Let $\Pi S$ be a parity reversed spinorial representation of $\Spin(V,q)$. We define the \emph{super translation algebra} on $V\oplus \Pi S$ by setting
\begin{equation}
[v,w] = 0, \quad [v,s] = 0,\quad [s,t] = \Gamma(s,t)
\end{equation}
for $v,w\in V$ and $s,t\in\Pi S$.
\end{definition}
The \emph{super Poincaré algebra} is given by the semidirect product of Lie algebras $\mathfrak p = (V\oplus \Pi S)\rtimes\mathfrak{spin}(V,q)$.
The preceding algebraic relations capture infinitesimal displacements of fermionic and bosonic particles. In what follows, we define the global extension of these relations by imposing distinguished spin structures on superspaces.
\begin{definition}
An \emph{almost spinorial superspace} is a superspace $M$, where
\begin{itemize}
    
    \item the structure group of the frame bundle is reduced according to
\begin{equation}
Q^F\to Q^B\times_M Q^F \to \Fr(TM)
\end{equation}
factoring through the double cover
\begin{equation}
\Spin(V,q)\xrightarrow{\alpha\times\id} \SO(V,q) \times \Spin (V,q)\to \GL(m|n),
\end{equation}
\item the bosonic distribution $B\subset TM$ has typical fibre $V$ such that $B\cong Q^B\times_{\SO(V,q)}V$, and
\item the spinorial distribution $F\subset TM$ has typical fibre $\Pi S$, where $S$ is a real spinorial representation of $\Spin(V,q)$, such that $F \cong Q^F \times_{\Spin(V,q)}\Pi S$.
\end{itemize}
For brevity, we shall write $Q^F \eqqcolon Q$.
\end{definition} 

Note that the Frobenius obstructions can be made explicit via
\begin{equation}\label{eq:frob-explicit}
f_F: \wedge^2 F \xrightarrow{[-,-]} TM \xrightarrow{\pr_B} B,
\end{equation}
as the orthogonal projection onto $B$ identifies $TM/F \cong B$.

The approach to superspaces via the bundle of frames shows that Frobenius obstructions to $\mathcal B$ and$\mathcal F$  must arise as maps of $\Spin(V,q)$ representations. We illustrate this for the integrability conditions of $\mathcal F$. To do so, we pick any local section $p$ of $Q = Q^ F$ over $U\subset M$. Recall that we have a map $\hat\alpha:Q^F \to Q^B$, such that $p$ gives rise to the following isomorphisms
\begin{equation}
p:\Pi S \times U \xrightarrow{\sim}{} \left.F \right\vert_U,\quad
\alpha(p): V \times U\xrightarrow{\sim}{}  \left.B \right\vert_U.
\end{equation}
Using the inverse of the maps induced by the section, we see that $f_ F$ descends fibrewise to a map of $\Spin(V,q)$ representations, as the following diagram commutes by construction:
\begin{equation}
\begin{tikzcd}
\wedge^2( F|_U) \arrow{r}{f_ F} \arrow{d}{p^{-1}\wedge p^{-1}}&
 B|_U\arrow{d}{\hat\alpha(p)^{-1}}\\
U\times\wedge^2(\Pi S)
\arrow{r}{\hat\alpha(p)^{-1}(f)}&
U\times V.
\end{tikzcd}
\end{equation}
As all local sections $p$ are related by a gauge transformation $g\in C^\infty(U,\Spin(V,q))$ and the preceding result is independent of $g$, it follows that the map $\hat\alpha(p)^{-1} (f_F)$ is $\Spin(V,q)$-equivariant, and can thus be identified with an element
\begin{equation}\label{eq:frob-2-form}
C_F \in\Omega^2(Q,V)^{\Spin(V,q)}_{\hor}.
\end{equation}
Note that equivalently, $C_F$ can be constructed by using the hom-tensor adjunction \begin{equation}
\Hom(\wedge^2(\mathcal F),\mathcal B) \cong \wedge^2(\mathcal F)^\vee\otimes_{\mathcal O _M}\mathcal B.
\end{equation}
Alternatively, using a $\Spin(V,q)$ connection $\varphi\in\mathcal A (Q)$, one can lift sections of $\mathcal F$ horizontally to $\mathcal T _Q$ and use the $\Spin(V,q)$-equivariance of the lifted sections.

\begin{lemma}\label{lemma:b-integrable}
The bosonic distribution is integrable along any purely even submanifold $\iota:\check M \to M$ whenever $\underline{\check M} = \underline M$ and $\underline\iota = \id_{\underline M}$. 
\end{lemma}
\begin{proof}
Note first that $\diff\iota: T\check M \to \iota^*TM$ corestricts to an isomorphism 
\begin{equation}
T\check M \cong \iota^*B\cong B_{\iota(\check M)}.
\end{equation}
For integrability, it suffices to show that the Frobenius obstruction $f_ B \coloneqq \pr_ F ([-,-]|_{\mathcal B \wedge  \mathcal B})$ vanishes. To that end, note that the obstruction presented in this way is a graded skew $\mathcal O _M$-bilinear map $\mathcal B \wedge \mathcal B \to \mathcal F$. As illustrated above, $f_ B$ corresponds fibrewise to a bilinear $\Spin(V,q)$-equivariant map $V\otimes V\to S$. Suppose the map was non-vanishing for some $v\otimes w\in V\otimes V$. Then, $-1\in\Spin(V,q)$ is in the kernel of the action on $V\otimes V$, whereas it changes the sign on $S$ -- a contradiction. Hence, the Frobenius obstruction vanishes. 
\end{proof}
The above lemma will be crucial when we identify ordinary spacetime manifolds within superspaces.
\begin{definition}
An almost spinorial superspace is said to be of \emph{type $\mathcal N = k$} where $k$ is the number of irreducible real spinorial representations $S$ decomposes into, that is,
\begin{equation}
S \cong \bigoplus_{i=1}^k S_0
\end{equation}
where $S_0$ is irreducible.
 If the irreducible spinorial representation consists of two semi spinor representations $S_0= S^+_0\oplus S^-_0$, the superspace is said to be of \emph{type $(\mathcal N^+,\mathcal N^-) = (k^+,k^-)$} if
 \begin{equation}
 S \cong \bigoplus_{i=1}^{k^+} S_0^+ \oplus\bigoplus_{i=1}^{k^-} S_0^-.
 \end{equation}
\end{definition}
There are variants of this definition, some authors use $\mathcal N$ to refer to the number of irreducible complex representations $S$ turns into upon complexification.
\begin{lemma}
If the almost spinorial superspace is of type $\mathcal N = 1$, the Frobenius map $C_F$ from equation~\eqref{eq:frob-2-form} coincides with the morphism induced by the map of representations $\Gamma:S\otimes S \to V$.
\end{lemma}
\begin{proof}
Note first that the Frobenius pairing is indeed $\mathcal O _M$-bilinear, such that by the associated bundle construction and by using the soldering form, it lifts to a $\Spin(V,q)$-equivariant map $\Sym^2(S)\to V$, which is non-degenerate. These maps are unique (up to rescaling), and described by theorem~\ref{thm:gamma-real}.
\end{proof}
\begin{definition}
An almost spinorial superspace is called \emph{spinorial superspace} if the Frobenius map $C_F$ coincides with the map induced by the morphism of $\Spin(V,q)$ representations $\Gamma: S\otimes S\to V$ from theorem~\ref{thm:gamma-real-unique}.
\end{definition}
Thus in particular, all $\mathcal N=1$ almost spinorial superspaces are spinorial superspaces.
Our definition of spinorial superspaces is motivated by geometric and physical observations: The spinors arise to capture geometric effects of $\Spin(V,q)$ that are not captured by $\SO(V,q)$. The compatibility condition ensures that the interplay between bosonic and spinorial fields arising from the underlying geometry of the superspace adheres to the structure of spin representations and globalises the notion of the \emph{super translation algebra}, as introduced in definition~\ref{def:super-translation}.\par
If $\mathcal N>1$, the underlying spinorial representation decomposes as $S= W \otimes S_0$ for an irreducible representation $S_0$ (if $S_0$ is graded, so is $W$, and the graded tensor product has to be taken). In that case, there exists a $\GL(W)$-symmetry in the theory, obtained by acting on the twisting space $W$. This is referred to in physics literature as an occurence of \emph{$R$-symmetry}. The physical nature of these symmetries is discussed in reference~\cite{figueroa2001busstepp}.\par
Henceforth, we shall exclusively consider spinorial superspaces, and embrace the fact that spinorial superspaces encode their spinorial nature in the spinorial distribution $\mathcal F \subset\mathcal T _M$. Furthermore, inspired by physics literature, we will also use $\Gamma$ to refer to the Frobenius map $C_F$.  
\begin{lemma}
A spinorial superspace is canonically equipped with
\begin{itemize}
\item a $\mathcal O _M$-bilinear metric $g$ on the bosonic distribution $\mathcal B\subset\mathcal T _M$,
\item a non-degenerate graded $\mathcal O _M$-symmetric or graded $\mathcal O _M$-skew symmetric form on the spinorial distribution $\mathcal F\to M$ or both, depending on the dimension of $V$, the signature of $q$, and $\mathcal N$.
\end{itemize}
Furthermore, the map $\Gamma:\wedge^2_{\mathcal O _M}\mathcal{F}\to \mathcal B$ is non-degenerate in the sense that for non-isotropic $v\in\mathcal B$, the bilinear form $g(\Gamma(-,-),v)$ is non-degenerate.
\end{lemma}
\begin{proof}
The partial pseudo-Riemannian metric on $\mathcal B\subset\mathcal T _M$ is clear in view of the $\SO(V,q)$-structure on the bosonic distribution. The other structures follow from properties of real spinorial representations and descend to the fibres by the associated bundle construction.
\end{proof}
If $S$ is of quaternionic type, there exist both symmetric and skew symmetric forms. If $S$ is of real or complex type, the non-degenerate bilinear form on $\mathcal F$ is graded symmetric if $\lfloor\dim(V)/2\rfloor =0,1\mod 4$, and graded skew else (keep in mind that $\mathcal F$ is oddly generated). In particular, if $d=3,4$, there exists an invariant graded symmetric bilinear form $g_\mathcal F$ on $\mathcal F$, such that $g\coloneqq g_\mathcal B\oplus g_\mathcal F$ defines a super Riemannian (or super Lorentzian) metric on $M$.

The representation theory of the $\Spin(V,q)$ group provides existence and uniqueness results for certain dual $\Gamma$ maps:
\begin{lemma}\label{lemma:gamma-superspace}
    Let $M$ be a spinorial superspace. There exists a unique morphism of vector bundles $\tilde\Gamma:\wedge^2_{\mathcal O _M} \mathcal F^\vee \to \mathcal B$ such that the duals $\gamma:\mathcal B\to \Hom_{\mathcal O _M} (\mathcal F^\vee,\mathcal F)$ and $\tilde\gamma:\mathcal B\to \Hom_{\mathcal O _M} (\mathcal F, \mathcal F^\vee)$ obey the Clifford relations
    \begin{align}
        \gamma(V)\circ\tilde\gamma(W)
        + (-1)^{p(V)p(W)}
        \gamma(W)\circ\tilde\gamma(V)
        &= 2 g(V,W)\cdot\id_\mathcal F,\\
        \tilde\gamma(V)\circ\gamma(W)
        + (-1)^{p(V)p(W)}
        \tilde\gamma(W)\circ\gamma(V)
        &= 2g(V,W)\cdot\id_{\mathcal F^\vee}.
    \end{align}
\end{lemma}
\begin{proof}
    By definition, the $\Gamma$ maps are non-degenerate. Lifting the discussion to the bundle of frames by means of the soldering form, the desired properties now follow directly from theorem~\ref{thm:gamma-real-unique}.
\end{proof}
We can generalise the notion of spinorial superspaces slightly: Recall that the irreducible spinorial representations are have a type $\mathbb F\in\{\mathbb R,\mathbb C,\mathbb H\}$ which depends on the signature of $q$. The division algebra $\mathbb F$ is isomorphic to the space of $\Spin(V,q)$-invariant automorphisms of $S$. Thus, depending on the type of $S$, there exists a compact group of $\Spin(V,q)$-invariant automorphisms of $S$, which we will denote by $K^\mathbb F=\{x\in\mathbb F : |x|=1\}$. Thus, the extended automorphism group of $S$ is
\begin{equation}
\Spin^\mathbb F(V,q)\coloneqq\left( \Spin(V,q)\times K^\mathbb F\right) / \{\pm 1\}.
\end{equation}
Note that there is a map $\alpha:\Spin^\mathbb F (V,q)\to \SO(V,q)$, constructed by modding out $K$ and subsequently taking the double cover map $\Spin(V,q)\to \SO(V,q)$ (this map is well-defined since it is independent of the choice of representative in $\Spin(V,q)$). This map fits in the short exact sequence
\begin{equation}
\begin{tikzcd}
1\arrow{r}&K\arrow{r}&\Spin^\mathbb F (V,q)\arrow{r}{\alpha}&\SO(V,q)\arrow{r}&1.
\end{tikzcd}
\end{equation}
\begin{definition}
A \emph{$\Spin^\mathbb F(V,q)$ superspace} is a superspace $M$, infinitesimally modelled by $V\times\Pi S$, together with a reduction of the structure group
\begin{equation}
\hat Q\to Q^ B \times_M \hat Q \to \Fr(TM)
\end{equation}
factoring through
\begin{equation}
\Spin^\mathbb F (V,q)\xrightarrow{\alpha\times\id} \SO(V,q) \times \Spin^\mathbb F (V,q)\to \GL(m|n).
\end{equation}
\end{definition}
In particular, the $\mathcal O_M$-bilinear map $\pr_B ([-,-])$ respects the reduction, such that it can again be viewed as the $\Gamma$ map
\begin{equation}
\Gamma\in\Omega^{0|2}(\hat Q,V)^{\Spin^\mathbb F (V,q)}_{\hor}.
\end{equation}
The generalisation to $\Spin^\mathbb F (V,q)$ structures is a reasonable mathematical generalisation, based on the observation that $\Spin^\mathbb F(V,q)$ is the full structure group of the fibres. Furthermore, some manifolds that do not admit $\Spin(V,q)$ structures do in fact admit $\Spin^\mathbb F (V,q)$ structures, see, e.g., reference~\cite{albanese2021spinh}.
Note that whenever $S$ is of real type, the notions of spinorial superspaces and $\Spin^\mathbb F(V,q)$ superspaces coincide.
\footnote{Note further that the latter two of the corresponding groups $K\in\{\Z_2, \operatorname{U}(1), \SU(2)\}$ appear as structure groups of electroweak force; a possible relation is however beyond the scope of this work.}
\begin{example} Every spinorial superspace can be turned into a $\Spin^\mathbb F(V,q)$ superspace, by setting
\begin{equation}
\hat Q = \left(Q \times_M K^\mathbb F\right) / \{\pm 1\}\to M.
\end{equation}
Note that this definition suggests a canonical trivialisation on the second factor. However, treating $K^\mathbb F$ as a true symmetry of the resulting theories, physics should not predict a choice of trivialisation. Thus, it is better to think of the second factor as a trivialisable (but not trivialised!) principal $K^\mathbb F$-bundle $M\times K^\mathbb F \coloneqq Q^\mathbb F\to M$, and of the $\Spin^\mathbb F(V,q)$-structure as 
\begin{equation}
\hat Q = \left(Q \times_{M} Q^\mathbb F \right) / \{\pm 1\} \to M
\end{equation}
The right action is given by $R_{[x,y]}([q,r]) = [R_x(q),\bar{y}r]$. It is standard to check that the $\Spin^\mathbb F(V,q)$ action is free on $\hat Q$, and transitive on the fibres. 
\end{example}
The compact group $K^\mathbb F$ is an internal symmetry of the spinorial representation, and is usually treated as an occurrence of \emph{$R$-symmetry}. For a detailed discussion of the physical interpretation of $R$-symmetries, see~\cite{figueroa2001busstepp}.\par
Moving forward, we shall endow any spinorial superspace with the $\Spin^\mathbb F (V,q)$ structure described above, but forget about the trivialisation of the $R$-symmetry arising from the second factor as physics predicts no canonical choice.

\begin{theorem}
Let $Q\to M$ be a spinorial superspace, together with a $\Spin^\mathbb F(V,q)$-structure $\hat Q = \left(Q\times_{M} Q^\mathbb F\right)/\{\pm 1\}$. There exists a canonical isomorphism
\begin{equation}
\mathcal A _{\hat Q} \cong \Omega^{2|0}(\hat Q,V)^{\Spin^\mathbb F (V,q)}_{\hor}\times\mathcal A _{Q^\mathbb F},
\end{equation}
where the first factor arises as the torsion component corresponding to $V$ along the bosonic distribution.\label{thm:superspace-connection}
\end{theorem}
\begin{proof}
Connections on $Q\times_M Q^\mathbb F$ are in bijective correspondence with connections on $\hat Q$ as the former discretely covers the latter. In particular, their adjoint bundles are canonically isomorphic. We now construct the isomorphism
\begin{equation}
\mathcal{A}_{Q} \xrightarrow{\sim}
 \Omega^{2|0}(\hat Q,V)^{\Spin^\mathbb F (V,q)}\cong 
 \Omega^{2|0}(Q,V)^{\SO(V,q)}.
\end{equation}
To that end, note that we have a morphism of principal bundles $Q\to Q^ B$, along the double covering $A:\Spin(V,q)\to \SO(V,q)$. Since the covering is discrete, $A_* : \spin(V,q)\to \mathfrak{so}(V,q)$ is an isomorphism of Lie algebras. Thus, connections on $Q^B$ isomorphically lift to connections on $Q$.
Note that $\mathcal A_{Q^B}$ is affine over $\Omega^1(Q^ B,\mathfrak{so}(V,q))^{\SO(V,q)}_{\hor}$. Recall that the torsion of a connection can be interpreted via the exact sequence of $\mathcal O_M$-modules
    \begin{equation}
        \begin{tikzcd}
            0\arrow{r}&\hat K\arrow{r}&\Omega^{1}(Q^ B,\mathfrak{so}(V,q))^{\SO(V,q)}_{\hor}\arrow{r}{\hat \alpha}&\Omega^2(Q^B,V)^{\SO(V,q)}_{\hor}\arrow{r}&\hat C\arrow{r}&0,
        \end{tikzcd}
    \end{equation}
    where $\hat\alpha (a) = \rho_*(a)\wedge \theta_B$, and $\rho_*:\mathfrak{so}(V,q)\to\End(V)$ is the induced map of Lie algebras from the $\SO(V,q)$-structure $Q^B$. $\hat K$ and $\hat C$ denote kernel and cokernel of $\hat\alpha$, respectively. This sequence arises from vector bundles associated to $Q^B$ via
    \begin{equation}
        \begin{tikzcd}
            0\arrow{r}&K\arrow{r}&V^\vee\otimes\mathfrak{so}(V,q) \arrow{r}{\alpha} &\wedge^2(V^\vee)\otimes V\arrow{r}&C\arrow{r}&0,
        \end{tikzcd}
    \end{equation}
    where $\alpha$ denotes the embedding
    \begin{equation}
        \id_{V^\vee}\otimes\rho_*:V^\vee\otimes\mathfrak{so}(V,q)\to V^\vee\otimes\End(V) = V^\vee\otimes V^\vee\otimes V
    \end{equation}
    with subsequent antisymmetrisation on the first two factors. Clearly, $\alpha$ is an isomorphism (and in particular, $\hat K = \hat C = \{0\}$). Thus, so is $\hat \alpha$; and since $\mathcal A _{Q^B}$ is affine over $\Omega^1(Q^B,\mathfrak{so}(V,q))^{\SO(V,q)}_{\hor}$, the map
    \begin{equation}
        \mathcal A _{Q^B}\to \Omega^2(Q^B,V)^{\SO(V,q)}_{\hor},\quad \tilde\varphi\mapsto\diff^{\tilde\varphi}\theta_\mathcal B
    \end{equation}
    is an isomorphism.\par
    Pulling back to $Q$ and subsequently restricting to bosonic forms yields the isomorphism
    \begin{equation}
        \Psi:\mathcal A _Q\to \Omega^{2|0}(Q,V)^{\SO(V,q)}_{\hor},
        \quad \varphi\mapsto
        (\widetilde{\pr}_{\mathcal B})^*\diff^\varphi\theta_\mathcal B.
    \end{equation}
\end{proof}
\begin{definition}
We say that a connection $\varphi^s\in\mathcal A _Q$ is a \emph{spinorial connection} if the image of its torsion in $\Omega^{2|0}(Q,V)^{\Spin(V,q)}_{\hor}$ vanishes, that is, $\widetilde{\pr}_\mathcal B^*\diff^{\varphi^s} \theta^\mathcal B=0$. The space of spinorial connections will be denoted by $\mathcal A^s_Q$.
\end{definition}
Note that the space of spinorial connections $\mathcal A^s_Q$ contains those connections that mimic the Levi-Civita connection along the bosonic distribution. A precise relation will be presented in section~\ref{section:restriction}.
\begin{corollary}
Let $Q\to M$ be a spinorial superspace, together with a $\Spin^\mathbb F(V,q)$-structure $\hat Q =\left( Q\times_{M} Q^\mathbb F\right) / \{\pm 1\}$. It holds
\begin{equation}
\mathcal A^s_Q \cong \mathcal A_{Q^\mathbb F}
\end{equation}
canonically.
\end{corollary}
\begin{proof}
This is just a special case of theorem~\ref{thm:superspace-connection}, restricted to the subspace of $\mathcal A _Q$ with vanishing image in $\Omega^{2|0}(Q,V)^{\Spin(V,q)}_{\hor}$.
\end{proof}
\begin{lemma}\label{lemma:superspace-berezinian}
If the spinorial representation underlying the superspace structure of $M$ comes with a symplectic form, there exists a canonical generator of the Berezinian sheaf $\Ber_M$.
\end{lemma}
\begin{proof}
Note that the $\Spin(V,q)$ action preserves the Berezinian: On $\mathcal F$, the determinant (and hence its inverse -- the Berezinian) is preserved by invariance of the symplectic form, and on $\mathcal B$, the structure group is already reduced to $\SO(V,q)\subset\SL(V)$, which also preserves the determinant.
\end{proof}
This canonical generator of $\Ber_M$ will be denoted by
\begin{equation}
\vol\in\Ber_M.
\end{equation}
If the invariant form of the spinorial representation is symmetric, one has to work projectively. The Berezinian bundle may be twisted by an orientation bundle to obtain the bundle of densities, as discussed in~\cite{deligne1999quantum}. For the examples we have in mind, we shall not require this generalisation.

\section{The Chirality Decomposition}\label{section:chiral-decomposition}
We shall now consider the scenario where under complexification, the spinorial distribution decomposes into chiral and antichiral components, i.e.,
\begin{equation}
\mathcal F^\C = \mathcal F ^{1,0}\oplus \mathcal F ^{0,1},
\end{equation}
where the distributions $\mathcal F^{1,0}$ and $\mathcal F^{0,1}$ correspond to the two holomorphic structures on the complexified spin representation $S^\C = S^{1,0}\oplus S^{0,1}$. We make use of the language presented in appendix~\ref{appendix:cs-manifolds}.\par
By using the approach via the bundle of frames, the nature of this decomposition can be expressed algebraically and occurs if and only if the underlying spin representation is of complex type. In this case, the structure group of $\mathcal F^\C$ is $\Spin^\C(V,q)$ and embeds canonically into $\GL_\C(S^\C)$. Since the complex nature of the CS manifold $(M,\mathcal O_M \otimes \mathbb C)$ only enters along the odd directions (and is therefore to be regarded as a completely algebraic property of the structure sheaf), the almost complex structure can be viewed as a complex structure along $\mathcal F^\C$. This means that $\mathcal F^{1,0}$ and $\mathcal F^{0,1}$ can be interpreted as holomorphic and antiholomorphic distributions, respectively.\par
We have the Dolbeault-like multigrading decomposition of the de Rham complex
\begin{equation}
\Omega^{p|q,r}(M) = \wedge^p(\mathcal B^\C)^\vee\otimes_{\mathcal O_M} \wedge^q(\mathcal F^{1,0})^\vee\otimes_{\mathcal O_M}\wedge^r(\mathcal F^{0,1})^\vee.
\end{equation}
It follows that the differential on functions decomposes as
\begin{gather}
\diff: \mathcal O _M\otimes \C \to \Omega^1(M,\mathbb C) = (\mathcal B^\C)^\vee \oplus (\mathcal F^{1,0})^\vee \oplus (\mathcal F^{0,1})^\vee,\\
\diff = \diff_\mathcal B\oplus \partial \oplus \bar\partial.
\end{gather}
Due to the non-integrability conditions on $\mathcal F$, we do not try to further decompose the exterior differential on forms into holomorphic and antiholomorphic parts. Rather, we consider a
spinorial connection $\varphi^s \in\mathcal A _{\hat Q}$. Since $\varphi^s$ is a $\Spin^\C(V,q)$ connection, and $S^\C = S^{1,0}\oplus S^{0,1}$ as $\Spin^\C(V,q)$ representations,
such a spinorial connection preserves the bosonic resp.~spinorial distributions. The corresponding covariant exterior derivative thus decomposes globally as
\begin{equation}
\diff^{\varphi^s} =
\diff^{\varphi^s}_B\oplus \partial^{\varphi^s} \oplus \bar\partial^{\varphi^s}
\end{equation}
on fibre bundles over $M$ arising from the decomposition $\mathcal T _M^\C = \mathcal B^\C \oplus \mathcal F^{0,1} \oplus \mathcal F^{0,1}$. By this notation, we mean that
\begin{equation}
\partial^{\varphi^s} : \Omega^\bullet(M,F^{1,0}) \to \Omega^{\bullet + 1}(M,F^{1,0}),
\end{equation}
and likewise for the other components.
The suggestive notation is chosen due to the analogy with the Dolbeault operators on complex manifolds.
\begin{definition}
We define the sheaf of \emph{chiral functions}
\begin{equation}
\mathcal O ^\omega_M \coloneqq \{f\in\mathcal O _M \otimes \C : Xf = 0 \textrm{ for all } X\in\Omega^0(M, \mathcal F^{0,1})\},
\end{equation}
and the sheaf of \emph{antichiral functions}
\begin{equation}
\mathcal O ^{\bar\omega}_M \coloneqq \{f\in\mathcal O _M \otimes \C : Xf = 0 \textrm{ for all } X\in\Omega^0(M, \mathcal F^{1,0})\}.
\end{equation}
\end{definition}
It turns out that these sheaves define CS submanifolds of $(\underline M, \mathcal O _M \otimes\C)$:
\begin{lemma}\label{lemma:chiral-integrable}
Upon complexification of the structure sheaf, $M$ gives rise to the two CS submanifolds $M^{1|0} = (\underline M , \mathcal O _M^\omega)$ and $M^{0|1} = (\underline M, \mathcal O _M^{\bar\omega})$.
\end{lemma}
\begin{proof}
It is clear that both $M^{1|0}$ and $M^{0|1}$ constitute locally ringed spaces. We prove now that $M^{1|0}$ is a CS manifold (the case of $M^{0|1}$ is analogous). By lemma~\ref{lemma:CS-submanifold}, it suffices to check that the distribution $\mathcal B^\mathbb C \oplus \mathcal F^{1,0}$ is integrable. By Frobenius' theorem, this amounts to verifying that the Frobenius pairing vanishes, that is, the $\mathcal F^{0,1}$ components of the commutator are zero. As in the discussion of Equation~\eqref{eq:frob-explicit}, the Frobenius obstructions of the kind $\pr_{0|1} ([-,-])$ can be lifted to the bundle of frames, where they take the form of symmetric morphisms of $\Spin(V,q)$ representations $S^{1,0}\otimes S^{1,0} \to S^{0,1}$, $V^\mathbb C\otimes V^\mathbb C\to S^{0,1}$, and $V^\mathbb C\otimes S^{1,0} \to S^{0,1}$. However, these morphisms are all identically zero:
\begin{itemize}
\item Suppose $S^{1,0}\otimes S^{1,0} \to S^{0,1}$ is non-zero for some $s\otimes t\in S^{1,0}\otimes S^{1,0}$. Acting by $-1\in\Spin(V,q)$, the left hand side remains unaltered, while the right hand side changes sign -- a contradiction.
\item A similar reasoning holds for the contribution from $V^\mathbb C\otimes V^\mathbb C\to S^{0,1}$, as $-1\in\Spin(V,q)$ is in the kernel of the representation on $V^\mathbb C$.
\item By theorem~\ref{thm:gamma-real} and representation theory of the spin group, the only symmetric map of representations dualises to $V^\mathbb C \otimes S^{1,0}\to S^{1,0}$, hence, $V^\mathbb C\otimes S^{1,0} \to S^{0,1}$ is identically zero.
\end{itemize}
Thus, the Frobenius obstruction vanishes, such that lemma~\ref{lemma:CS-submanifold} applies.
\end{proof}
We generalise the notion of (anti)chiral functions to that of \emph{(anti)chiral fields}:
\begin{definition}
Let $E\to M$ be a complex vector bundle associated to $Q$. We say that a section $\chi\in\Omega^0(M,E)$ is \emph{chiral} if $\bar\partial^{\varphi^s}\chi = 0$, and \emph{antichiral} if $\partial^{\varphi^s}\chi = 0$. We will write $\Gamma^\omega(M,E)$ for the sheaf of chiral sections and $\Gamma^{\bar\omega}(M,E)$ for the space of antichiral sections.
\end{definition}
The relation to the sheaf of (anti)chiral functions is as follows:
\begin{lemma}
Let $E\to M$ be a complex vector bundle associated to $Q$. 
\begin{itemize}
\item The sheaf of chiral sections $\Gamma^\omega(M,E)$ is a locally free $\mathcal O ^{\omega}_M$-module.
\item The sheaf of antichiral sections $\Gamma^{\bar\omega}(M,E)$ is a locally free $\mathcal O ^{\bar\omega}_M$-module.
\end{itemize}
\end{lemma}
\begin{proof}
The sheaves are locally free as we are working with locally trivial vector bundles and because the complex structure gives non-degenerate algebraic relations on the sections that smoothly depend on $M$. To prove the first claim, let $X\in\Gamma^\omega(M,E)$ be a chiral section and let $f\in \mathcal O^\omega_M$ be a chiral function. We compute
\begin{equation}
\bar\partial^{\varphi^s}(fX) = \bar\partial^{\varphi^s}(f)X + f\bar\partial^{\varphi^s} X = 0,
\end{equation}
proving that $fX\in\Gamma^\omega(M,E)$. The antichiral case is analogous.
\end{proof}
We shall dentote by $\Ber_M^{1,0}$ and $\Ber_M^{0,1}$ the Berezinian sheaves of the CS manifolds $M^{1,0}$ and $M^{0,1}$. Note that $\mathcal F^{1,0}$ and $\mathcal F^{0,1}$ are infinitesimally modelled by irreducible complex spinorial representations. They thus carry non-degene\-rate invariant bilinear forms, and an analogue of lemma~\ref{lemma:superspace-berezinian} holds:
\begin{corollary}
If the spinorial representation $S^{1,0}$ underlying $\mathcal F^{1,0}$ carries an invariant non-degene\-rate graded symmetric bilinear form, the CS manifold $M^{1,0}$ carries a canonical volume form, denoted $\vol^{1,0}\in\Ber_M^{1,0}$. The analogous statement holds for $M^{0,1}$.
\end{corollary}
Again, if the form is symmetric, one has to work projectively. However, we will not require this in our treatment.
In the $d=4$ super Yang-Mills setup, we will encounter the decomposition of the superfield of interest into chiral and antichiral components. 

\section{Restriction to Ordinary Spacetimes}\label{section:restriction}
We shall now discuss the restriction of the structures to the embedding of a purely even spacetime manifold $i:\check M\to M$. Physically, this corresponds to identifying an 
ordinary spacetime manifold within superspace. Upon restricting to the embedding, fields defined over $M$ decompose into component fields defined on $\check M$.
\begin{definition}
An \emph{underlying ordinary spacetime} of the superspace $M$ is a purely even integral submanifold $i:\check M \to M$ such that:
\begin{itemize}
\item The topological spaces of $\check M$ and $M$ coincide, and the map of topological spaces $\underline i : \underline{\check M}\to \underline M$ is the identity.
\item $\check M$ integrates the bosonic distribution, that is,
\begin{equation}
\mathcal B |_{i(\check M)} = \mathcal T _{\check M}.
\end{equation}
\end{itemize}
\end{definition}
We shall identify $\check M$ with its image in $M$, thus writing $\mathcal B |_{\check M}$ instead of $\mathcal B |_{i(\check M)}$ et cetera. Note that such embeddings always exist due to lemma~\ref{lemma:b-integrable}. Fixing an embedding of this kind induces a plethora of structures on the ordinary spacetime manifold $\check M$. We shall discuss the relevant ones below.
\begin{lemma}
Let $i:\check M \to M$ be an embedding of an underlying ordinary spacetime manifold $\check M$ into the spinorial superspace $M$. The superspace structure on $M$ induces a $\Spin^\mathbb F(V,q)$ structure and a compatible Riemannian structure on $\check M$.
\end{lemma}
\begin{proof}
Note that upon pulling back by the embedding $i$, the bundle of frames decomposes into
\begin{equation}
i^*Q^B\times_{\check M} i^*Q^F \to \check M.
\end{equation}
The first factor defines the Riemannian structure, and the latter one defines the $\Spin^\mathbb F(V,q)$ structure. By definition of the superspace structure, there exists a map of principal bundles $Q^F\to Q^B$
with kernel $K$, and pulling back by $i^*$, we obtain the compatible structures
\begin{equation}
i^*Q^F\to i^*Q^B \to \check M.
\end{equation}
\end{proof}
Note that the pullback $i^*Q^F$, together with the spinorial representation $S$, allows for the construction of a spinor bundle over the ordinary spacetime
\begin{equation}
S\check M = i^*Q^ F\times_{\Spin^\mathbb F(V,q)} \Pi S\to \check M.
\end{equation}
We have the following isomorphisms:
\begin{lemma}\label{lemma:distr-iso}
The corestriction $\psi_T:\mathcal T _{\check M}\to i^*\mathcal B$ of $\diff i:\mathcal T_{\check M} \to i^*\mathcal T_M$ is an isomorphism. Furthermore, there exists a canonical isomorphism $\psi_S:\Omega^0(\check M, S\check M) \to i^*\mathcal F$.
\end{lemma}
\begin{proof}
Upon pulling $\mathcal T _M$ back by $i$, it decomposes into $i^*\mathcal B\oplus i^*\mathcal F$. By definition, $i^*\mathcal B = \mathcal B|_{\check M}$ coincides with $\mathcal T _{\check M}$.\par
Upon pulling back the soldering form component $\theta^\mathcal F$ by $i$, the desired isomorphism for $\mathcal F|_{\check M}$ with $\Omega^0(\check M , S\check M)$ is given explicitly by
\begin{equation}
\Omega^0(\check M, S\check M) \xrightarrow{\sim}\Omega^0(i^*Q,\Pi S)^{\Spin(V,q)}
\xrightarrow{\langle i^*\theta^\mathcal F,\cdot\rangle} i^*\mathcal F.
\end{equation}
\end{proof}
Finally, we have the following relations between the spinorial connection on $M$ and the spin connection on $i^*Q^ F \to \check M$:
\begin{lemma}\label{lemma:restriction-levi-civita}
Let $\varphi^s\in\mathcal A^s_Q$ be a spinorial connection. The pullback connection $i^*\varphi^s\in\mathcal A _{i^*Q}$ is a $\Spin^\mathbb F(V,q)$-connection on $i^*Q=\check Q$ (and thus in particular on the spinor bundle $S\check M$); and the pullback $i^*\left((\varphi^s)^\mathcal B\right)\in\mathcal A _{i^*Q^B}$ of the $Q^B$-part of $\varphi^s$ is the Levi-Civita connection on $(\check M, g|_{\check M})$.
\end{lemma}
\begin{proof}
It is clear that the pullback of the connection is a connection on the pullback bundle. Note that on spinorial connections, the $2|0$-component of the $V$-part of the torsion $\Theta^\varphi$ vanishes. Upon pulling back to $\check M$, this is the only part that remains -- hence, $i^*\left((\varphi^s)^\mathcal B\right)$ is torsion free and metric, and thus, the Levi-Civita connection.
\end{proof}

\section{Split Superspaces}\label{sec:split-superspaces}
In this section, we first construct a class of spinorial superspaces that arise from spinor bundles over ordinary Riemannian spin manifolds, and then give relevant formulas for computing integrals on these superspaces.
\subsection{The Construction}
Throughout this section, we shall denote by $(\check M,\check g)$ an ordinary (i.e., non-super) pseudo-Riemannian manifold with infinitesimal model $(V,q)$. Further, $S$ shall be a real spinorial representation of $\Spin(V,q)$. Let us further assume that $\check M$ is spin --  that is, there exists a morphism of principal bundles
\begin{equation}
\begin{tikzcd}
\check Q \arrow{r}&\Fr^{\SO}(T\check M)\arrow{r}&\check M,
\end{tikzcd}
\end{equation}
factoring through the structure groups $\Spin(V,q)\to\SO(V,q)$. Here, $\check \pi: \check Q\to\check M$ is the lift of the bundle of orthonormal frames $\Fr^{\SO}(T\check M)\to \check M$.
\begin{definition} A \emph{split superspace} is a spinorial superspace modelled on a split supermanifold of the type
\begin{equation}
M\coloneqq \Pi S\check M,
\end{equation}
where
\begin{equation}
\tau:\Pi S\check M = \check Q \times_{\Spin(V,q)} \Pi S \to \check M
\end{equation}
is the parity reversed spinor bundle associated to the spin structure $\check Q$ and the spinor module $S$.
\end{definition}
We shall now construct the spinorial sheaf $\mathcal F$ that endows the split superspace with a spinorial structure as defined in section~\ref{sec:spin-structures}.
\begin{lemma}
The Levi-Civita connection on $\Fr^{\SO}(T\check M)$ lifts canonically and uniquely to the spin structure $\check Q \to \check M$.
\end{lemma}
\begin{proof}
The morphism of principal bundle $\psi:\check Q \to \Fr(T\check M)$ is equivariant along the discrete cover of structure groups $\Spin(V,q)\to\SO(V,q)$. In particular, the Lie algebras of the structure groups are canonically isomorphic -- hence, the induced morphism on equivariant Lie-algebra-valued forms (in particular the connections) is an isomorphism.
\end{proof}
We shall denote this spin connection on $\check Q$ by $\varphi\in\mathcal A _{\check Q}$. It gives rise to a splitting of the short exact sequence
\begin{equation}
\begin{tikzcd}
0\arrow{r}& V\check Q \arrow{r}&
T\check Q\arrow{r}{\diff\check \pi}
\arrow[bend right=60,swap]{l}{\varphi}&\check\pi^* T\check M \arrow{r}&0.
\end{tikzcd}
\end{equation}
We use the canonical embedding of the underlying ordinary spacetime into the superspace by means of the zero section, which we shall denote by
\begin{equation}
i:\check M \to M.
\end{equation}
We denote by $Q\to M = \Pi S\check M$ the frame bundle of the superspace, where we assume without loss of generality that the structure group is already reduced to $\Spin(V,q)$. 
\begin{lemma}
There exists an isomorphism of principal $\Spin(V,q)$-bundles over $\Pi S\check M$
\begin{equation}
\psi: Q\to \check Q \times \Pi S,
\end{equation}
where the $\Spin(V,q)$-action on the right hand side is given by
\begin{equation}
R'_g(q,s) = (qg, g^{-1}s),\quad g\in\Spin(V,q),q\in\check Q, s\in\Pi S,
\end{equation}
and the projection $\check Q \times \Pi S \to \Pi S\check M$ by the quotient of the diagonal action of $\Spin(V,q)$, i.e.,
\begin{equation}
\pi' : \check Q \times \Pi S \to \Pi S\check M, \quad (q,s)\mapsto [(q,s)], \quad (q,s)\sim (q.g,g^{-1}s).
\end{equation}
\end{lemma}
Note that the latter factor is a \emph{vector space} instead of the spinor bundle.
\begin{proof}
Note that $Q=\tau^*\check Q$ as the entire $\Spin(V,q)$-structure is already captured in $\check Q$.
By definition of the categorical pullback, $\tau^*\check Q$ fits in the commutative diagram
\begin{equation}
\begin{tikzcd}
\tau^*\check Q \arrow{r}\arrow{d}& \check Q\arrow{d}\\
\Pi S\check M\arrow{r}{\tau}&\check M
\end{tikzcd},
\end{equation}
and thus coincides with the fibre product over $\check M$
\begin{align}
\check Q \times_{\check M} \Pi S M &\cong \check Q \times_{\check M} \big(\check Q \times_{\Spin(V,q)} \Pi S \big).
\end{align}
The desired isomorphism is given explicitly by
\begin{equation}
\psi:\check Q \times_{\check M} (\check Q \times_{\Spin(V,q)}\Pi S M) \to \check Q \times\Pi S,\quad
(q,[(q',s')]) = (q,[(q,s)])\mapsto (q,s).
\end{equation}
Thus, we choose the first factor of $\check Q$ to pick the suitable representative of the associated bundle.
By equivariance, the right action on $\check Q\times_{\check M}\Pi S \check M$ must be compatible with $\psi$, such that
\begin{equation}
R_g ((q,[(q',s')]) = (q.g,[(q,s)]) = (q.g, [(q.g, g^{-1}s]) \mapsto (q.g, g^{-1}s) =: R'_g(q,s).
\end{equation}
It follows that the projection 
\begin{equation}
\check Q \times \Pi S \to \check Q \times_{\Spin(V,q)}\Pi S = \Pi S\check M,\quad 
(q,s)\mapsto [(q,s)]
\end{equation}
is invariant under the right $\Spin(V,q)$-action $R'$ on $\check Q \times \Pi S$ and thus well-defined. Summarising, we have the following commutative diagram
\begin{equation}
\begin{tikzcd}
\check Q \times_{\check M}\Pi S\check M \times\Spin(V,q)\arrow{r}{R}\arrow{d}{\psi\times\id}&
\check Q \times_{\check M}\Pi S\check M \arrow{r}\arrow{d}{\psi}& \Pi S\check M.\\
\check Q \times \Pi S \times\Spin(V,q)\arrow{r}{R'}&\check Q \times \Pi S \arrow[swap]{ru}{\mod \Spin(V,q)}&
\end{tikzcd}
\end{equation}
This concludes the proof as each morphism of principal bundles of same structure group and over the same base is automatically an isomorphism.
\end{proof}
In what follows, we construct the spinorial distribution $\mathcal F\subset \mathcal T _M$. A good candidate would be the vertical tangent bundle of the parity reversed spinor bundle $\tau^*\Pi S\check M$ as its sections contain all odd generators of $\Der(\mathcal O_M)$. However, the vertical tangent bundle is integrable as the vertical vector fields are closed under the bracket. The strategy we follow is to deform the vertical bundle.
Recall from the preceding proof that the frame bundle of $M=\Pi S\check M$ is given by $\check Q\times\Pi S$. It fits into the commutative diagram
\begin{equation}
\begin{tikzcd}
\check Q \times \Pi S\arrow{r}{\pr_{\check Q}}\arrow{d}{q} & \check Q \arrow{d}{\pi}\\
\Pi S\check M \arrow{r}{\tau} & \check M.
\end{tikzcd}
\end{equation}
Here, $q$ denotes the quotient by the diagonal action $R_g(\check q, s) = (\check q.g,g^{-1}.s)$.
The spin connection $\varphi\in\mathcal A_{\check Q}$ induces a splitting $X\mapsto \nabla_X^\varphi$ on the tangent space of the parity reversed spinor bundle. To see this, consider the projections
\begin{equation}
\begin{tikzcd}
\check Q & \arrow[swap]{l}{\pr_{\check Q}}  \check Q \times \Pi S \arrow{r}{\pr_{\Pi S}}&\Pi S,
\end{tikzcd}
\end{equation} 
and the following diagram of associated vector bundles
\begin{equation}\label{eq:split-tangent-diagram}
\begin{tikzcd}
0\arrow{r}&\check Q \times T\Pi S\arrow{r}\arrow{dd}{\pr_{\Pi S}^* \diff q}
&T\check Q \times T\check \Pi S\arrow{r}\arrow{dd}{\diff q}&T\check Q \times\Pi S\arrow[swap]{d}{\diff\pi\times\id}\arrow{r}& 0\\
&&&\pi^*T\check M\times\Pi S\arrow{d}{q}\arrow[bend right=60,swap]{u}{\varphi\times\id}&\\
0\arrow{r}&\tau^*\Pi S\check M \arrow{r}&T\Pi S\check M \arrow{r}{\diff\tau}&
\tau^*T\check M\arrow{r}\arrow[bend left=60]{l}{\nabla^\varphi}&0.
\end{tikzcd}
\end{equation}
Here, the rows are exact sequences of vector bundles over $\check Q \times \Pi S$ and $\Pi S\check M$, respectively. The upshot of this point of view is that the upper sequence has a canonical splitting. Note that the canonical isomorphism $\Omega^\bullet(Q, W)^{\Spin(V,q)}_{\hor}\cong \Omega^\bullet(M, \check Q \times_{\Spin(V,q)} W)$ can be used to lift sections of associated bundles in the lower sequence to equivariant sections in the upper one. 
\begin{definition}
The \emph{Euler vector field} on a vector space $S$ is the image $\mathcal E$ of the identity map $\id_S$ under the canonical identification of $S^\vee \otimes S$ with the space of vector fields with linear coefficients. In particular, a choice of basis of $S$ induces coordinate functions $\{s^i\}$, such that 
\begin{equation}
\mathcal E_S = \sum_i s^i \pdiff{}{s^i}.
\end{equation}
\end{definition}
The first description is clearly $\GL(S)$-invariant -- alternatively, one may check that the coordinate representation of $\mathcal E _S$ is invariant under linear transformations. We can thus define such vector fields on fibres of associated vector bundles as follows:
\begin{definition}
The Euler vector field $\mathcal E$ on $\Pi S\check M$ is the image of $\mathcal E _S$ under the quotient by the $\Spin(V,q)$-action
\begin{equation}
\check Q\times T\Pi S \xrightarrow{\pr_{\Pi S}^*\diff q} \tau^* \Pi S\check M\into T\Pi S \check M.
\end{equation}
\end{definition}
Clearly, $\mathcal E$ is vertical, as can be seen with help of the diagram in equation~\eqref{eq:split-tangent-diagram}.\par
We define the deformation of the vertical tangent bundle of $\Pi S\check M$ by the monomorphism of vector bundles
\begin{equation}
\delta_\mathcal F:\tau^*\Pi S\check M \to T\Pi S\check M, \quad X\mapsto X + \frac{1}{2}\nabla^\varphi_{\Gamma(\mathcal E, X)},
\end{equation}
where $\Gamma: \tau^*\Pi S\check M \otimes \tau^*\Pi S\check M \to \tau^*T\check M$ is constructed via the map of representations $\Gamma: S\otimes S \to V$. Further, we define the regular distribution
\begin{equation}
\mathcal F\coloneqq \Omega^0\left(\Pi S\check M, \delta_\mathcal F(\tau^*\Pi S\check M)\right)\to \mathcal T _{\Pi S\check M}.
\end{equation}
\begin{lemma}
$\mathcal F$ is a spinorial distribution on $M=\Pi S\check M$, and turns $M$ into a spinorial superspace.
\end{lemma}
\begin{proof}
We need to show that upon restricting to the body, the commutator coincides with the $\Gamma$ map of representations. To that end, we compute
\begin{equation}\label{eq:commutator-relation-proof}
\left.[\delta_\mathcal F (X),\delta_\mathcal F (Y)]\right\vert_{\check M} = 
\left.[X+\frac{1}{2}\nabla_{\Gamma(\mathcal E, X)}^\varphi,Y+\frac{1}{2}\nabla_{\Gamma(\mathcal E, Y)}^\varphi]\right\vert_{\check M}
\end{equation}
for $X,Y\in\Omega^0(\Pi S\check M, \tau^*\Pi S\check M)$.
We lift the vertical vector fields $X,Y$ to $\Spin(V,q)$-invariant vector fields on $Q=\check Q\times\Pi S$ according to the diagram in equation~\eqref{eq:split-tangent-diagram}.
We may equip $S$ with a basis $\{f_\mu\}$, inducing coordinate functions $\{s^\mu\}$, and it follows that 
\begin{equation}
f_\mu = \pdiff{}{s^\mu} \eqqcolon \partial_\mu
\end{equation}
is a coordinate frame.
Note that $\Spin(V,q)$ acts linearly on $\Pi S$, such that coordinate frames are taken to coordinate frames under the $\Spin(V,q)$ action.
We first evaluate equation~\eqref{eq:commutator-relation-proof} on constituents of the coordinate frame.
The expression for the commutator simplifies as coordinate vector fields commute; hence,
\begin{align}\label{eq:commutator-relation-2}
\left[\partial_\mu+\frac{1}{2}\nabla_{\Gamma(\mathcal E, \partial_\mu)}^\varphi,\partial_\nu+\frac{1}{2}\nabla_{\Gamma(\mathcal E, \partial_\nu)}^\varphi\right]= 
\frac{1}{2}[\partial_\mu,\nabla_{\Gamma(\mathcal E, \partial_\nu)}^\varphi] + 
\frac{1}{2}[\nabla_{\Gamma(\mathcal E, \partial_\mu)}^\varphi,\partial_\nu] + 
\frac{1}{4}[\nabla_{\Gamma(\mathcal E,\partial_\mu)}^\varphi,\nabla_{\Gamma(\mathcal E,\partial_\nu)}^\varphi].
\end{align}
Recall that using coordinate frames on $\Pi S$, the Euler vector field can be written as
\begin{equation}
\mathcal E = \sum_\lambda s^\lambda \partial_\lambda.
\end{equation}
Now, we can evaluate the individual terms of equation~\eqref{eq:commutator-relation-2}. Note that upon restricting to $\check M$, we are only interested in terms that are non-vanishing along the inclusion. In particular, the coordinates $s^\mu$ vanish, such that
\begin{align}
[\partial_\mu,\frac{1}{2}\nabla_{\Gamma(\mathcal E, \partial_\nu)}^\varphi] &= \frac{1}{2}
\sum_\lambda [\partial_\mu, s^\lambda \nabla_{\Gamma(\partial_\lambda, \partial_\nu)}^\varphi]\\
&= \frac{1}{2}\nabla_{\Gamma(\partial_\mu,\partial_\nu)} \mod \Pi S.
\end{align}
Upon restricting to the body, it follows
\begin{equation}
[\partial_\mu,\nabla_{\Gamma(\mathcal E, \partial_\nu)}^\varphi]|_{\check M} =(\psi_S^*\Gamma)(\partial_\mu|_{\check{M}}, \partial_\nu |_{\check{M}})
\end{equation}
as after restriction, the structure sheaf is reduced to ordinary functions on $\check M$, where the covariant derivative reduces to the ordinary differential.
The term involving $[\nabla_{\Gamma(\mathcal E,\partial_\mu)}^\varphi,\nabla_{\Gamma(\mathcal E,\partial_\nu)}^\varphi]$ vanishes along the body as it has coefficients that are linear in $\Pi S$.
Putting everything together, we find
\begin{equation}
\left.[\delta_\mathcal F(\partial_\mu), \delta_\mathcal F(\partial_\nu)]\right\vert_{\check M} = (\psi_S^*\Gamma)(\partial_\mu|_{\check M}, \partial_\nu|_{\check M}),
\end{equation}
proving the claim.
\end{proof}
To implement the functor-of-points approach, all we need to do is to impose an auxiliary base $A$, which we will choose to be of the form $A=(\pt, \wedge^\bullet(\mathbb R^a)^\vee)$. Thus, we view the trivialised fibration
\begin{equation}
M\times A\xrightarrow{\pr_A} A
\end{equation}
as a spinorial superspace, parametrised by the auxiliary base $A$.\par
\begin{remark}
    Due to Batchelor's theorem~\cite{batchelor1984graded}, any \emph{trivial} family of supermanifolds, i.e., a supermanifold over a base $A$ with the submersion given by $M\times A \xrightarrow{\pr_A} A$, is non-canonically isomorphic to a split one. It can be shown that in this case, all spinorial superspaces are non-canonically isomorphic to split ones. \par
    An analogous statement might hold true in the category of supermanifolds that only contain \emph{submersed} bases $A$ by replacing the topological space $\underline M$ with an underlying ordinary spacetime manifold as discussed in the preceding section. However, such a `relative version' of Batchelor's theorem remains to be worked out.
\end{remark}
As a concrete example, let us apply the preceding notions to an affine space.
\begin{example}[Affine Superspace]
    Let $(\check M,\check g)$ be the affine space modelled by $(V,q)$. In this case, the preceding construction becomes particularly simple: Note that we have a canonical trivialisation of the tangent bundle
    \begin{equation}
        T\check M \cong \check M\times V\xrightarrow{\pr}\check M,
    \end{equation}
    inducing the trivialisation of the bundle of orthogonal frames
    \begin{equation}
        \Fr^{\SO(V,q)}(T\check M) \cong \check M\times\SO(V,q)\xrightarrow{\pr} M.
    \end{equation}
    The lift $\Spin(V,q)\to \SO(V,q)$ induces the corresponding spin structure $\check Q\to \check M$. We choose a spinorial representation $S$ of $\Spin(V,q)$, and obtain the supermanifold
    \begin{equation}
        M=\check M\times\Pi S.
    \end{equation}
    The functions on $M$ are all of the form
    \begin{equation}
        f = \sum_I f_I \xi^I\in \mathcal O _M,\quad f_I\in C^\infty(\check M),\quad \xi^I\in\wedge^\bullet(S^\vee),
    \end{equation}
    where $I$ is a multiindex.
    The tangent bundle decomposes as
    \begin{equation}
        TM\cong M\times (V \oplus \Pi S)\to M.
    \end{equation}
    Let us fix (global) generators $\{\partial_i\}$ corresponding to an orthonormal basis $\{e_i\}$ of $V$, and $\{\partial_\mu\}$ corresponding to a basis $\{f_\mu\}$ of $\Pi S$. Let ${\Gamma_{\mu\nu}^i} e_i= \Gamma(f_\mu,f_\nu)$. In these coordinates, the structure distribution is given by
    \begin{equation}
        \mathcal F=\vspan_{\mathcal O _M}\big\{D_\mu = \partial_\mu + \frac12\xi^\nu{\Gamma_{\mu\nu}}^i\partial_i:\mu=1,\dots,\dim(S)\big\},
    \end{equation}
    such that $[D_\mu,D_\nu] ={\Gamma_{\mu\nu}}^i\partial_i$. This affine superspace constitutes a local model for general spinorial superspaces.
\end{example}
The case where the signature of $q$ is $(1,m-1)$ is referred to as \emph{super Minkowski space}, and is the running example in physics literature, see, for example,~\cite{castellani1991supergravity}.
\subsection{Computing Integrals}
In local spin frames on the total space $\Pi S\check M$, that is, sections of $Q=\check Q \times \Pi S \to \Pi S\check M$, we have
\begin{equation}
q=(\tau^*(e_1,\dots,e_m),\delta_\mathcal F (\partial_1)^\vee,\dots,\delta_\mathcal F (\partial_n)^\vee),
\end{equation}
where the $(e_i)$ are a local frame on $\check M$, i.e., a local section of $\check Q\to \check M$. In what follows, we identify frames and coframes by dualising the bases of the respective distributions.
Note that there exists a $\GL(m|n)$ transformation of unit Berezinian relating $q$ with a frame that has  the $\partial_\mu$ as odd generators (rather than $\delta_\mathcal F(\partial _\mu$), given explicitly by
\begin{equation}
\Psi : (\tau^*(e_1,\dots,e_n),\partial_1,\dots,\partial_n)\mapsto (e_1,\dots,e_n,\delta_\mathcal F (\partial_1),\dots,\delta_\mathcal F (\partial_n)),
\end{equation}
or, as a matrix,
\begin{equation}
\Psi = 
\left(\begin{array}{ c | c }
    \id & \Gamma(\mathcal E,\cdot ) \\
    \hline
    0 & \id
\end{array}\right).
\end{equation}
In particular, it holds:
\begin{corollary}
The Berezinian section associated to a spin frame $q$ is the same as the one associated to $q.\Psi^{-1}$ upon embedding $Q$ into the bundle of $\GL(m|n)$ frames.
\end{corollary}
Since $i^*\tau^*(e_1,\dots,e_m) = (\tau\circ i)^*(e_1,\dots,e_m) = (e_1,\dots,e_m)$, the usual formula
\begin{equation}
\int _M f[q] = \int_{\check M} i^*(\partial_1\dots\partial_n f) e^1\wedge\dots\wedge e^m
\end{equation}
remains valid. Furthermore, note that under the pullback by $i$, we may -- after antisymmetrising -- replace the coordinate derivations by the deformed ones again if the odd dimension of the supermanifold is equal to two:
\begin{lemma}
Let $n=2$ be the odd dimension of the superspace. It holds
\begin{equation}
i^*\left(\partial_1 \partial_2 f\right) =
 \frac{1}{2}i^*\left( (\delta_\mathcal F (\partial_1)\circ\delta_\mathcal F ( \partial_2) - 
 \delta_\mathcal F (\partial_2)\circ\delta_\mathcal F ( \partial_1)) f\right).
\end{equation}
\end{lemma}
\begin{proof}
It holds
\begin{equation}
i^*(\delta_\mathcal F(\partial_1)\circ \delta_\mathcal F(\partial_2) f) =
i^*(\partial_1\circ(\partial_2 + \frac{1}{2}\nabla^\varphi_{\Gamma(\mathcal E, \partial_2)} f)=
i^*\left(\left(\partial_1\circ\partial_2 + \frac{1}{2}\nabla^{\varphi}_{\Gamma(\partial_1,\partial_2)}\right) f\right).
\end{equation}
Upon antisymmetrising, the second term vanishes due to the symmetry properties of $\Gamma$. Finally, recall that for coordinate frames, $\partial_1\partial_2 = -\partial_2\partial_1$, concluding the proof.
\end{proof}
To summarise, we have the following theorem:
\begin{theorem}\label{thm:superspace-integration}
Let $n=2$ be the odd dimension of the superspace, and let us denote by $q = (\tau^*(e^1,\dots,e^m),\delta_{\mathcal F} (\partial_1),\delta_{\mathcal F} (\partial_2))$ a spin frame. Let $\vol = [q]$ be the canonical superspace volume form, and let $\check\vol = e^1\wedge\dots\wedge e^m$ be the volume form on the ordinary spacetime. It holds
\begin{equation}
\int_M f\vol = \frac{1}{2}\int_{\check M} i^*\left(\left(\delta_\mathcal F(\partial_1)\circ \delta_\mathcal F(\partial_2) - \delta_\mathcal F(\partial_1)\circ \delta_\mathcal F(\partial_2)\right)f\right)\check\vol.
\end{equation}
\end{theorem}
If $n\ge 4$, correction terms involving the Riemannian curvature of $(\check M, \check g)$ will appear. However, in our treatment of super Yang-Mills theories, the odd dimensions of the superspaces we consider will always be equal to two.

\section{Super Yang-Mills Theories in Dimensions 3 and 4}\label{section:sym}
We finally have all the required tools to invariantly define pure super Yang-Mills (sYM) theory on spinorial superspaces. We first give a general overview of the expected component fields, and then explicitly treat the $\mathcal N = 1$ case in $d=3$ and $d=4$ spacetime dimensions in Minkowski signature. We reduce the superspace Lagrangians to underlying ordinary spacetimes $i:\check M \to M$ in each of the cases and discuss the component fields that appear due to the reduction. In both cases, the reduced Lagrangian is of the form
\begin{equation}
    \check{\mathcal L} (a,\lambda) = -\frac{1}{2}\left(\mu\|F^a\|^2 - \langle\lambda,\Dirac^a\lambda\rangle\right) \check\vol,
\end{equation}
where $\mu=1/2$ in $d=3$ and $\mu=1$ in $d=4$.
Here, $\check\vol$ is the Riemannian volume form on $\check M$, $\lambda$ denotes a real $\ad(i^*P)$-valued spinor field, and $a = i^*\mathsf a$ is a connection on the principal bundle $\check P = i^*P\to\check M$ restricted to the ordinary spacetime. Thus, the reduced Lagrangian is a minimally coupled system of Dirac and Yang-Mills Lagrangians.\par
The upshot of the superspace formulations is that the Lagrangians we consider are manifestly gauge invariant. Furthermore, both the Yang-Mills and the Dirac component fields emerge from a single component of the superspace curvature, emphasising the duality between the spinor field and the \emph{field strength tensor} (the curvature of the reduced connection).\par
We will first define the component fields of the theory. These correspond to certain components of the gauge connection and its curvature. Then, we will give superspace formulations of the super Yang-Mills theories in $d=3$ and $d=4$ and show how to reduce them to an underlying ordinary spacetime manifold. We will restrict the procedure to the case where the family of supermanifolds is trivial, i.e., of the form $M\times A \to A$.
\subsection{The Component Fields}
To define super Yang-Mills theories, we consider some fixed principal $G$-bundle $P\to M$ over the superspace. Here, $G$ is an ordinary Lie group. As the form of the superspace Lagrangian will depend on the dimension and signature of the infinitesimal model $(V,q)$ of the ordinary spacetime, we shall first discuss the fields that arise upon restricting the super Yang-Mills fields to the ordinary spacetime via the embedding $i:\check M\to M$. The space of relevant fields will be given by the gauge constrained connections
\begin{equation}
\mathfrak F^{sYM} = \mathfrak F_P = \{\mathsf a \in\mathcal A _P : \langle\mathcal F\wedge\mathcal F,F^\mathsf a\rangle = 0\},
\end{equation}
that is, connections that are flat along the spinorial distribution $\mathcal F \subset\mathcal T _M$. Here,
\begin{equation}
F^\mathsf a = \diff\mathsf a + \frac{1}{2}[\mathsf a \wedge \mathsf a]\in\Omega^2(P,\mathfrak g)) \cong \Omega^2(M,\ad(P))
\end{equation}
is the curvature of the connection $\mathsf a$. Due to the gauge constraint, it decomposes according to
\begin{equation}
F^\mathsf a = f^\mathsf a \oplus \phi^\mathsf a \in 
\Omega^{2|0}(P,\mathfrak g)^{G}_{\hor}\oplus
\Omega^{1|1}(P,\mathfrak g)^{G}_{\hor}
\end{equation}
as the $0|2$-component vanishes. We can dualise the $\tilde\Gamma$ map to give rise to a morphism $\tilde\gamma:\mathcal F^\vee\otimes \mathcal B^\vee\to \mathcal F$, allowing us to produce an $\ad(P)$-valued section of the spinorial distribution $\mathcal F$ by means of the duality pairing
\begin{equation}
\psi^\mathsf a\coloneqq\frac{1}{\sqrt n} \langle \phi , \tilde\gamma\rangle
\in \Omega^0(M,\ad(P))\otimes\mathcal F^\vee \cong
\Omega^0(Q\times_M P, \Pi S\otimes \mathfrak g)^{\Spin(V,q)\times G}_{\hor}.
\end{equation}
This superspace spinor will play a central role in the $d=3$ and $d=4$ super Yang-Mills theories.
We now restrict the theory to the ordinary spacetime via the pullback $i^*P\eqqcolon \check P \to \check M$ and obtain the following component fields of the theory:
\begin{definition}
The \emph{component fields} of the sYM theory are given by
\begin{itemize}
\item an ordinary connection $a=i^*\mathsf a \in\mathcal A _{\check P}$,
\item an $\ad(\check P)$-valued spinor field $\lambda = (\psi_S^{-1}\otimes \id) (\psi^\mathsf a|_{\check M})$,
\end{itemize}
where $\psi_S :\Omega^0(\check M, S\check M) \to i^*\mathcal F$ is the isomorphism from lemma~\ref{lemma:distr-iso}.
\end{definition}
These component fields will constitute the relevant fields upon reducing the theory to ordinary spacetime manifolds, and are the only fields appearing in the $d=3$ and $d=4$ sYM theory. In higher dimensions and different signatures, there may occur further component fields.

\subsection{\texorpdfstring{Lorentzian $d=3$ sYM}{Lorentzian d=3 sYM}}
The irreducible real spinorial representation $S$ in $d=3$ with Minkowskian signature is two-dimensional and of real type. The invariant bilinear form extends $g_\mathcal B$ to a super-Riemannian metric $g=g_\mathcal B\oplus g_\mathcal F$ on $TM$. We give the superspace Lagrangian, suitably adapted from~\cite{deligne1999quantum}:
\begin{definition}
The $d=3$ Lorentzian superspace Lagrangian is defined by
\begin{equation}
\mathcal L : \mathfrak F^{\textup{sYM}} \to \Ber_M,\quad
\mathsf a \mapsto \frac{1}{2}\|\psi^\mathsf a\|^2_{\kappa \otimes \epsilon}\vol_g,
\end{equation}
where $\kappa$ denotes the Killing form on $\mathfrak g$, and $\epsilon = g_\mathcal F$ the restriction of the super-Lorentzian metric to the spinorial distribution $\mathcal F\subset \mathcal T _M$.
\end{definition}
Note that this Lagrangian is defined directly on the superspace. Gauge-invariance is immediately clear since $F^\mathsf a$ is $G$-equivariant, and hence, so is $\psi^\mathsf a$, and the Killing form is $G$-invariant.\par
We now want to reduce the field theory to the underlying ordinary spacetime manifold:
\begin{theorem}
Upon restricting the theory to the embedded ordinary spacetime manifold, we obtain the Yang-Mills-Dirac Lagrangian
\begin{equation}
\check{\mathcal L} (a,\lambda) = \left(-\frac{1}{4}\|F^a\|^2 + \frac{1}{2}\langle\lambda,\Dirac^a\lambda\rangle\right)\check\vol,
\end{equation}
where the component fields are given by
\begin{gather}
a = i^*\mathsf a \in\mathcal A _{\check P},\quad \lambda = (\psi_S^{-1}\otimes\id)(\psi|_{\check M})\in\Omega^0(\check M, S\check{M}\otimes\ad(\check P)).
\end{gather}
\end{theorem}
\begin{proof}
Recall that the reduced Lagrangian in terms of local coordinates is obtained by taking derivatives in the odd directions and suitably adapting the volume form, see, e.g.,~\cite{kessler2018supergeometry}.
We would like to circumvent the use of coordinates on superspace and instead solely rely on coordinates on the relevant spinor module as well as $\Spin(V,q)$ frames on $M$. This bears the advantage that we can exploit the representation theory of the $\Spin(V,q)$ group and require calculus only on the vector space $S$.\par
We first sketch the general situation: Any Lagrangian $\mathcal L$ on superspace can be reduced to a Lagrangian $\check{\mathcal L}$ on the embedded ordinary manifold $\check M$ by `integrating out the odd directions', as discussed in references~\cite{deligne1999quantum, kessler2018supergeometry}.
At this point, recall that integration along the odd directions is purely algebraic. It is thus sensible to invoke Batchelor's theorem, together with the remark in section~\ref{sec:split-superspaces}, which tells us that the superspace $M$ is diffeomorphic to a split one of the kind $\Pi S\check M \to \check M$. As the integral is invariant under diffeomorphisms, it thus suffices to compute the claim on such a split superspace. Using the frame bundle, we can write
\begin{equation}
\Pi S\check M = \check Q \times_{\Spin(V,q)} \Pi S,
\end{equation}
arising as an associated bundle using the real irreducible representation $S$. Here, we identify $i^*Q = \check Q\to \check M$ as the spin structure on the body.
Following the discussion in section~\ref{sec:split-superspaces}, we can use lemma~\ref{thm:superspace-integration} to compute the reduced Lagrangian, and find
\begin{equation}
\check{\mathcal L} = \frac{1}{2}i^*\left(\sum_{X,Y}\tilde\epsilon(X^\vee,Y^\vee)X\circ Y \big((\epsilon\otimes\kappa)(\psi^\mathsf a,\psi^\mathsf a)\big)\right)\check{\vol}
\end{equation}
where $X=\delta_\mathcal F (\partial_1)$ and $Y=\delta_\mathcal F (\partial_2)$ for brevity.
Let us pick a spinorial superspace connection $\varphi^s\in\mathcal A_Q^s$ and consider the pullback bundle $P\times_M Q\to M$. Note that $\mathsf a \oplus \varphi\eqqcolon A$ defines a connection that is compatible with both $G$ and the spinorial structure. In what follows, we denote by $\nabla^A\eqqcolon \nabla$ the corresponding covariant derivative. Then
\begin{equation}
\frac{1}{2}X\circ Y \big((\epsilon\otimes\kappa)(\psi,\psi)\big) = X\big( (\epsilon\otimes\kappa)(\nabla_Y\psi,\psi)\big)
= (\epsilon\otimes\kappa)(\nabla_X\nabla_Y\psi,\psi) - (\epsilon\otimes\kappa)(\nabla_Y\psi,\nabla_X\psi).
\end{equation}
Upon pulling this function back to $\check M$, we obtain
\begin{equation}
\frac{1}{4}i^*\left(\sum_{X,Y} \tilde{\epsilon}(\tilde X,\tilde Y)\big((\epsilon\otimes\kappa)(\nabla_X\nabla_Y\psi,\psi) - (\epsilon\otimes\kappa)(\nabla_Y\psi,\nabla_X\psi)\big)\right) = \frac{1}{2}\langle\lambda,\Dirac^a\lambda\rangle - \frac{1}{4}\|F^a\|^2,
\end{equation}
where we made use of lemmas~\ref{lemma:spinor-curvature} and \ref{lemma:dirac-laplace} (see appendix).
\end{proof}
We thus found that upon reducing the theory to the body, the superspace formulation of super Yang-Mills theory reproduces the well-established Dirac-Yang-Mills-Lagrangian.
\subsection{\texorpdfstring{Lorentzian $d=4$ sYM}{Lorentzian d=4 sYM}}\label{subsection:sym-d-4}
The irreducible spinorial representation $S$ of $d=4$ Minkowski space is four-dimensional and of complex type. Upon complexification, it turns into two half-spin representations $S^\mathbb C = S^+\oplus S^-$. Alternatively, $S^\mathbb C$ can be decomposed into chiral and antichiral subspaces
\begin{equation}
S^\mathbb C = S^{1,0}\oplus S^{0,1},
\end{equation} 
obtained by the projections $\pr_{1,0} = \id - \i\otimes J$ and $\pr_{0,1} = \id + \i\otimes J$ as discussed in section~\ref{section:chiral-decomposition}.
The $\mathbb C$-module structure of $S^{1,0}\oplus S^{0,1}$ is thus given by $z.(s,t) = (zs,\bar z t)$, such that the invariant pairing $S^{1,0}\otimes S^{0,1} \to \mathbb C$ is invariant under $\operatorname{U}(1)$, and the non-degenerate bilinear form respects the decomposition $S^{1,0}\oplus S^{0,1}$.\par
Passing to the spinorial distribution via the construction of associated bundles, the complexified spinorial distribution $\mathcal F^\mathbb C$ decomposes into chiral and antichiral components $\mathcal F^\mathbb C = \mathcal F^{1,0}\oplus \mathcal F^{0,1}$, as discussed in section~\ref{section:chiral-decomposition}.
\begin{lemma}
Let $\mathcal O _M^\omega\subset \mathcal O _M \otimes \C$ be the sheaf of functions that are annihilated by $\mathcal F^{0,1}$, and let $\mathcal O _M^{\bar\omega}\subset \mathcal O _M \otimes \C$ be the sheaf of functions that are annihilated by $\mathcal F^{0,1}$. The locally ringed spaces $M^{1,0} = (\underline M,\mathcal O^\omega_M)$ and $M^{0,1} = (\underline M,\mathcal O^{\bar\omega}_M)$ are CS manifolds of dimensions $4|2$.
\end{lemma}
\begin{proof}
This follows directly from lemma~\ref{lemma:chiral-integrable}.
\end{proof}
In the same way, one verifies that the sheaf of sections annihilated by $\mathcal F^{1,0}$ defines a complementary CS manifold of dimension $4|2$.\par
Moving on, we shall work in the complexified description and take real parts. Since the complexified spinor $(\psi^\mathsf a)^\mathbb C$ arises from a real representation of $\Spin(V,q)$, we obtain a reality condition, which reads
\begin{equation}
(\psi^\mathsf a)^\mathbb C = 
\chi^\mathsf a\oplus \overline{\chi^\mathsf a}
\in\Omega^0\big(M,\ad(P)^\mathbb C\big) \otimes_{\mathcal O _M} \big((\mathcal F^{1,0})^\vee \oplus (\mathcal F^{0,1})^\vee\big).
\end{equation}
In this complexified picture, we also complexified the Lie algebra $\mathfrak g$ of the structure group $G$, and the tensor products are taken over $\mathbb C$. Note that the real form inside $\mathfrak g^\mathbb C$ is imaginary / skew-Hermitean for compact matrix Lie groups. We can now define the $d=4$ sYM Lagrangian:
\begin{definition}
The $d=4$ Lorentzian superspace sYM Lagrangian is given by
\begin{gather}
\mathcal L : \mathfrak F^{\textup{sYM}} \to\Ber_M^{1,0}\oplus\Ber_M^{0,1},\\
\mathcal L (\mathsf a) = \frac{1}{4}\big(\|\chi^\mathsf a\|^2_{\epsilon\otimes\kappa}\vol^{1,0} + \|\overline{\chi^\mathsf a}\|^2_{\overline{\epsilon\otimes\kappa}}\vol^{0,1}\big) = \frac{1}{2}\operatorname{Re}\big(\|\chi^\mathsf a\|^2_{\epsilon\otimes\kappa}\vol^{1,0}\big)
\end{gather}
where the chiral resp.~antichiral sections of the respective Berezinian bundles are interpreted as in section~\ref{section:chiral-decomposition}.
\end{definition}
As in the $d=3$ case, the Lagrangian is manifestly gauge invariant. We also have similar reduction results upon restricting the theory to an ordinary spacetime manifold:
\begin{theorem}
The $d=4$ Lorentzian sYM Lagrangian reproduces the component Lagrangian
\begin{equation}
\check{\mathcal L}(a,\psi,E) = -\frac{1}{2}\left(\|F^a\|^2 - \langle\lambda,\Dirac^a\lambda\rangle\right)\check{\vol},
\end{equation}
with the component fields
\begin{gather}
a = i^*\mathsf a \in\mathcal A _{\check P},\quad \lambda = (\psi_S^{-1}\otimes\id)(\psi|_{\check M})\in\Omega^0(\check M, S\check{M}\otimes\ad(\check P)).
\end{gather}
\end{theorem}
\begin{proof}
We continue working in the complexified picture.
The $1|1$-component $\phi^\mathsf a$ of the curvature decomposes into $(\mathcal B^\C)^\vee  \wedge (\mathcal F^+)^\vee $- and $ (\mathcal B^\C)^\vee  \wedge (\mathcal F^-)^\vee$-components, which we shall denote by $\phi^\mathsf a = \zeta\oplus\bar\zeta$. Bianchi's identity now takes the form
\begin{equation}
\diff^\mathsf a f + \diff^\mathsf a \zeta + \diff^\mathsf a \bar\zeta = 0.
\end{equation}
On the complexified spinor modules, the $\Gamma$ map is of the type $S^+\otimes S^-\to V$ -- thus,
\begin{equation}
\chi = \langle\bar\zeta,\tilde\gamma\rangle,\quad \bar\chi = \langle\zeta,\tilde\gamma\rangle.
\end{equation} 
 We shall compute the fibrewise integration of the chiral component $\epsilon(\chi,\chi)\vol$ along the odd directions -- the antichiral component will be its complex conjugate. As in the preceding chapter, we compute the integral fibre-wise on the split superspace. It holds
\begin{equation}
\check{\mathcal L} = \frac{1}{4}i^*\left(\epsilon^{XY} X\circ Y (\epsilon\otimes\kappa) (\chi,\chi)\right)\check\vol + c.c.,
\end{equation}
where $c.c.$ denotes the complex conjugate with respect to the canonical real form arising from the complexification, and $X,Y$ is a frame of the chiral distribution $\mathcal F ^+$. Summation over $X$ and $Y$ is implied. As in the $d=3$ case, we pick a spinorial connection on $\hat Q\to M$, and lift $X$ and $Y$ to horizontal vector fields on the pullback bundle $P\times_M \hat Q \to M$. We shall denote by $A=\mathsf a\oplus \varphi^s$ the connection on the latter bundle, and by $\nabla^A \eqqcolon \nabla$ the induced covariant derivative on the bundle $\ad (P)^\C\otimes\mathcal F^+$.
It holds
\begin{equation}
\left(\epsilon^{XY} X\circ Y \right)(\epsilon\otimes\kappa) (\chi,\chi) = -2\epsilon^{XY} \big( (\epsilon\otimes\kappa) (\nabla_X \chi, \nabla_Y \chi) - (\epsilon\otimes\kappa) (\nabla_X \nabla_Y\chi, \chi) \big).
\end{equation}
The first of the two terms on the right hand side is computed in lemma~\ref{lemma:d-4-curvature}, and its restriction to $\check M$ equates to $\|F^a\|^2$. The restriction to $\check M$ of the second term evaluates to 
\begin{equation}
\epsilon(\chi|_{\check M},\Dirac^a(\chi_{\check M}))).
\end{equation}
Recall that the Hermitean form intertwines $\mathcal F^{1,0}$ and $\mathcal F ^{0,1}$, such that we may identify the preceding equation as the duality pairing
\begin{equation}
\langle \overline{\chi}|_{\check M},\Dirac^a(\chi_{\check M}))\rangle.
\end{equation}
By adding the complex conjugate, we obtain
\begin{equation}
\check{\mathcal L} = -\frac{1}{2}\left( \|F^a\|^2 - \langle\lambda,\Dirac^a \lambda \rangle\right)\check\vol
\end{equation}
as the Hermitean pairing takes $\mathcal F^{1,0}\otimes_\mathbb R\mathcal F^{0,1}\to\mathcal O_M$.
\end{proof}
Note that in physics literature, one often finds an additional $\ad(\check P)\otimes\mathfrak u(1)$-valued component field $E$. This is due to the fact that in the corresponding computations, a musical isomorphism is used in-between to turn $\lambda$ into a one-form on superspace.
\section{Conclusion and Outlook}In this document, we established a mathematically precise and concise notion of spinorial superspaces
by means of $\Spin(V,q)$-structures on supermanifolds. We constructed a rather general framework
in which Lagrangian field theories on these superspaces can be reduced to the underlying ordinary
spacetime manifold, giving rise to component fields that are well-known in the physics literature.
Finally, we applied the results to $\mathcal N = 1$ super Yang-Mills theories in $d = 3$ and $d = 4$ spacetime
dimensions. Our treatment has the upshot of being coordinate-free and manifestly gauge invariant.
Together with the spinorial structure of the superspaces we consider, this allows for a rather geometric interpretation of the theories at hand. The coordinate-free description also allows for the
formulation of these theories on curved superspaces, generalising the established physics treatment
which is usually restricted to the case where the ordinary spacetime is chosen to be Minkowski
space, or a torus.
Further work has to be carried out to characterise the general supersymmetry properties of the
considered models when considered over arbitrary superspaces. In particular, a component field
characterisation of the supersymmetries arising from our framework of spinorial superspaces is
desirable.
It would also be of interest to construct similar super Yang-Mills theories in higher dimensions and
in different signatures. The Euclidean case is of particular interest in the mathematical community
in hopes of assigning invariants to the manifold by characterising the solution space of the theory.
To that end, it would be of interest to explore the structure of spinorial superspaces when the
underlying spinorial representation is of quaternionic type. In this way, a $d = 4$ Euclidean sYM
theory could be formulated on superspaces.
As our framework of spinorial superspaces is rather general, it is also feasible to consider other field
theories on them. In particular, the supersymmetric sigma model can likely be cast as a theory on
spinorial superspaces. On the other hand, one could construct more general field theories directly
on superspace, and explore the arising component fields as well as the reduction of the theory to
the underlying ordinary spacetime manifold.

\subsection*{Acknowledgements}
The author would like to thank Victor Pidstrygach, Greg Weiler and Marcel Bigorajski for their many comments and suggestions.
\printbibliography
\appendix
\section{Complex Supermanifolds}\label{appendix:cs-manifolds}
In what follows, we discuss the notion of \emph{complex supermanifolds}:
\begin{definition}
Let $M=(\underline M, \mathcal O_M)$ be a supermanifold. The associated \emph{complex supermanifold} (or \emph{CS manifold} for short) is the locally graded ringed space $M=(\underline M ,\mathcal O_M\otimes_\mathbb R \mathbb C)$ constructed by complexifying the structure sheaf.
\end{definition}
A few remarks are in order.
\begin{itemize}
\item The complex nature of this kind of manifold only enters the picture along the odd directions; and the underlying real supermanifold can be obtained by restricting the structure sheaf to its real part.
\item By construction, we only requested \emph{smoothness} of the functions instead of real or complex analycity. This is due to the fact that the complex behaviour of functions is merely algebraic and has no underlying topological implications.
\item There exists the more restrictive notion of \emph{complex analytic supermanifolds}, which is discussed in detail in reference~\cite{kessler2018supergeometry}. In this setting, one indeed requires the sheaf of functions to be locally isomorphic to the sheaf of holomorphic functions on $\mathbb C^{m|n}$. We shall not require this notion.
\end{itemize}
We can restrict the structure sheaf of a CS manifold to obtain a CS submanifold as follows:
\begin{lemma}\label{lemma:CS-submanifold}
Let $M$ be a CS manifold of dimension $m|n_1+n_2$ that admits a decomposition $TM = D_1\oplus D_2$ into regular distributions of rank $m|n_1$ and $0|n_2$, respectively. If $D_1$ is an integrable distribution, we obtain a CS manifold $M_1$ of dimension $m|n_1$, where $M_1$ is given by
\begin{equation}
M_1 = (\underline M, \mathcal O_{M_1}),\quad \mathcal O_{M_1} = \{s\in \mathcal O_M\otimes \C : Xs = 0 \textrm{ for all }X\in\Omega^0(M,D_2)\}.
\end{equation}
\end{lemma}
In other words, the structure sheaf arises by requiring that the functions are annihilated by $D_2$.
\begin{proof}
By Leibnitz' rule, it is clear that the restriction of the sheaf is a sheaf of locally graded subrings: Suppose that $X\in\Omega^0(M,D_1)$, and $s,t\in\mathcal O _{M_1}$. It follows
\begin{equation}
X(st) = (Xs) t + (-1)^{p(X)p(s)}s(Xt) = 0
\end{equation}
as each of the summands is zero. Therefore, $st\in\mathcal O _{M_1}$. Grading and locality follow from the corresponding properties of $\mathcal O_M$, and integrability of $D_1$ ensures the existence of the CS submanifold in question.
\end{proof}
Thus whenever a function on $M$ is actually contained in a suitable subsheaf $\mathcal O_{M_1}$, we can view it as a function defined on the CS submanifold $M_1$. We will exploit this fact to define the superspace Lagrangian of $d=4$ super Yang-Mills theory.

\section{Frame Computations}\label{appendix:frame-computations}
\subsection*{General Remarks}
The proofs of the following lemmas intertwine the superspace structure with Cartan's and Bianchi's identities. To economically account for these different geometric concepts, we shall resort to computations in local spin frames. We apologise in advance for the `battle of indices' that is about to take place.\par
We now fix the notation surrounding the musical isomorphisms in local frames. To that end, note that whenenver $X$ and $Y$ are odd vector fields, it holds $g(X,Y) = -g(Y,X)$ by the sign rule. Let $X,Y$ and $Z$ be sections of the spinorial distribution, that is, $X,Y,Z\in\mathcal F\subset\mathcal T _M$. Since $\mathcal F$ is purely odd, it suffices to establish the notation for odd local generators (for example, one could pick them such that they arise from spin frames). The general case follows by extending in a graded $\mathcal O _M$-multilinear fashion.\par
Let us write $\epsilon = g_\mathcal F$ to emphasize the signs occurring due to the grading. We shall use the second index of $\epsilon$ to raise indices
\begin{equation}
    A^\flat = \epsilon(\cdot,A) \Leftrightarrow \epsilon^{XY} A_Y = A^X.
\end{equation}
Since by definition $(A^\flat)^\sharp = A$, or equivalently $\epsilon^{XY}\epsilon_{YZ} = \delta^X_Z$, it follows 
\begin{equation}
    A_Z = \delta^X_Z A_X = \epsilon^{XY}\epsilon_{YZ} A_X = - \epsilon_{YZ} A^Y
\end{equation}
such that $\epsilon_{ZY}A^Y = A_Z$. Thus, we consistently use the second index of $\epsilon$ to both raise and lower indices. Since frame computations, especially these involving raising or lowering an even amount of indices, are rather prone to sign errors, we try avoiding the use of local frames and instead, resort to more conceptual methods whenever possible.
\begin{lemma}\label{lemma:f-cartan}
    Let $X,Y$ be horizontal lifts of $\tilde X, \tilde Y\in\mathcal F$. It holds
    \begin{equation}
        \iota_X\iota_Y\diff^\mathsf a \phi = -\iota_{[X,Y]}f
    \end{equation}
\end{lemma}
\begin{proof}
The proof is a little exercise involving Cartan's magic formula:
    \begin{align}
        \iota_X\iota_Y\diff^\mathsf a \phi &= -\iota_X\iota_Y\diff^\mathsf a f \\&=
        -\iota_X L_Y f \\&=
        -\iota_{[X,Y]} f,
    \end{align}
    where we used that $f$ vanishes on the spinorial vectors; and $\iota_X\mathsf a =\iota_Y\mathsf a = 0$ as $X$ and $Y$ are horizontal.
\end{proof}
In the following, $n$ shall denote the odd dimension of $M$ (or, equivalently, the rank of the spinorial distribution $\mathcal F\subset\mathcal T _ M)$.
\begin{lemma}
    Let $\psi$ be the sYM spinor associated to a connection $\mathsf a \in \mathcal A _P$, let $\varphi^s\in\mathcal A _{\hat Q}^s$ be a spinorial connection, and let $X$ be the horizontal lift of a section of the spinorial distribution $\mathcal F$ to $\mathcal T_P$. Set $\nabla^\mathsf a \eqqcolon \nabla$ and let $i:\check M \to M$ be the embedding of the ordinary spacetime manifold.
    It holds
    \begin{equation}
        (\nabla_X\psi)\big|_{\check M} =
        \frac{1}{\sqrt n}\langle f, \tilde\gamma\circ\gamma(X)\rangle\big|_{\check M}.
    \end{equation}
\end{lemma}
\begin{proof}
    We start off by plugging in the definition of the sYM spinor, such that
    \begin{align}
        \sqrt n \nabla_X \psi &=
        \iota_X\diff^\mathsf a \langle \phi,\tilde\gamma\rangle  
        =\langle\iota_X\diff^\mathsf a\phi, \tilde\gamma\rangle,
    \end{align}
    where we used that $\tilde\gamma$ is covariantly constant in the final step. We now expand $\tilde\gamma$ in a local frame
    \begin{equation}
        \tilde\gamma = \sum_{Y}\tilde\gamma(Y^\vee)\otimes Y
    \end{equation} 
    and make use of lemma \ref{lemma:f-cartan}, such that
    \begin{align}
        \sqrt n \nabla_X \psi &=
        \left\langle\iota_X\diff^\mathsf a \phi,\sum_Y\tilde\gamma(Y^\vee)\otimes Y\right\rangle\\
        &=\sum_Y \langle\iota_Y\iota_X \diff^\mathsf a \phi,\tilde\gamma(Y^\vee)\rangle\\
        &=-\sum_Y\langle\iota_{[X,Y]} f, \tilde\gamma(Y^\vee)\rangle.
    \end{align}
    Upon restricting to the underlying even manifold, we have $[X,Y]|_{\check M} = (\psi_S^*\Gamma)(X|_{\check M},Y_{\check M}))$, such that
    \begin{align}
        \sqrt n\cdot  (\psi_S^{-1}\otimes\id)\big((\nabla_X \psi)\big|_{\check M}\big) &= \sum_Y\left\langle f, \tilde\gamma(Y^\vee)\otimes\Gamma(X,Y)\right\rangle\Big|_{\check M}\\
        &=\langle f, \tilde\gamma\circ\gamma(X)\rangle\big|_{\check M},
    \end{align}
    concluding the proof.
\end{proof}
In the remainder of this appendix, we shall work out algebraic properties of component fields associated to a fixed connection $\mathsf a$ of the principal bundle $P\to M$ on the superspace. Therefore, we shall frequently write $F=F^\mathsf a$ for the curvature, as well as $f^\mathsf a = f$, $\psi^\mathsf a = \psi$ and $\phi^\mathsf a = \phi$ for the component fields.
\subsection*{\texorpdfstring{$d=3$ Frame Computations}{d=3 frame computations}}
We now collect a few lemmas that contain the key computations to deduce key properties of the sYM superspace Lagrangian in $d=3$. For notational convenience, we shall frequently identify the restricted spinorial distribution $i^*\mathcal F$ with a spinor bundle $S\check M\to\check M$ over the ordinary spacetime manifold.
\begin{lemma}\label{lemma:spinor-curvature}
    Let $\psi$ be the sYM spinor, let $i:\check M \to M$ be an embedding of an underlying even manifold, and let $i^*\mathsf a = a \in \mathcal A_{i^*P}$ be the pullback connection. It holds
    \begin{equation}
        i^*\big(\epsilon^{XY} (\epsilon\otimes\kappa)(\nabla_X\psi,\nabla_Y\psi)\big) = \|F^{a}\|^2,
    \end{equation}
    where $\kappa$ denotes the Killing form on $\mathfrak g$, and the norm of the curvature is taken with respect to $\kappa$ and the induced Riemannian metric $\check g$ on $\check M$.
\end{lemma}
\begin{proof}
    Recall first the Clifford identities
    \begin{align}
        \tr(\gamma\circ\tilde\gamma) = n\cdot g,\quad
        Q\coloneqq
        \tr(\tilde\gamma\circ\gamma
        \circ\tilde\gamma\circ\gamma)
        |_{\Sym^2(\Omega^{2|0}(M))}
        = - n\|\cdot\|^2_g,
    \end{align}
    where $n$ is the odd dimension of $M$, and compositions of the kind $\gamma\circ\tilde\gamma$ are interpreted as duality pairings of $S$ with $S^\vee$. With help of the preceding lemma, we find
    \begin{align}
        n\cdot i^*\big(\epsilon^{XY}(\epsilon\otimes\kappa)(\nabla_X\psi,\nabla_Y\psi)\big) &= 
        i^*\big(\epsilon^{XY}(\epsilon\otimes\kappa)(\langle f, \tilde\gamma\circ\gamma(X)\rangle,\langle f, \tilde\gamma\circ\gamma(Y)\rangle)\big)\\
        &= 
        i^*\big(\epsilon^{XY}\epsilon_{WZ}\kappa(\langle f, \tilde\gamma(W^\vee)\circ\gamma(X)\rangle,\langle f, \tilde\gamma(Z^\vee)\circ\gamma(Y)\rangle)\big),
    \end{align}
    where summation over $W,X,Y$ and $Z$ is implied.
    We make use of the identity $\epsilon^{XY}\epsilon_{WZ} = \delta^X_W\delta^Y_Z - \delta^X_Z\delta^Y_W$
    as well as $\sum_X \gamma(X)\circ\tilde\gamma(X^\vee) = \tr(\gamma\circ\tilde\gamma)$, such that
    \begin{equation}
        i^*\big(\delta^X_W\delta^Y_Z 
        \kappa(\langle f, \tilde\gamma(W^\vee)\circ\gamma(X)\rangle,\langle f, \tilde\gamma(Z^\vee)\circ\gamma(Y)\rangle)\big) = i^*\big(\kappa(\langle f,g\rangle,\langle f,g\rangle)\big) = 0
    \end{equation}
    as $g$ is symmetric and $i^*f$ is a two-form (and hence skew). Next, keeping in mind the symmetry properties of forms and the traces of $\gamma$ (resp. $\tilde\gamma$) maps,
    \begin{align}
    i^*\big(\delta^X_Z\delta^Y_W 
        \kappa(\langle f, \tilde\gamma(W^\vee)\circ\gamma(X)\rangle,\langle f, \tilde\gamma(Z^\vee)\circ\gamma(Y)\rangle)\big) &= i^*\big(\langle\kappa(f\wedge f),\tr(\tilde\gamma
        \circ\gamma\circ\tilde\gamma\circ\gamma)\rangle\big)\\
        &= n\|i^*f\|^2 = n\|F^{a}\|^2,
    \end{align}
    as claimed.
\end{proof}
Let us express the Dirac operator in terms of the language presented above:
\begin{lemma}
Let $\psi$ be the sYM spinor, and let $(\psi_S^{-1}\otimes\id)(\psi|_{\check M}) = \lambda$ be the
associated component field.
Denote by $a=i^*\mathsf a$ the pullback connection on the ordinary spacetime manifold. It holds
\begin{equation}
        \Dirac^{a}\lambda = \frac{1}{\sqrt n}\langle\diff^\mathsf a \phi, \tilde \gamma\circ\gamma\circ\tilde\epsilon\rangle\big|_{\check M},
\end{equation}
where we identified $i^*\mathcal F$ with the spinor bundle $S\check M \to \check M$.
\end{lemma}
\begin{proof}
    The Dirac operator, applied to a spinor, is defined as
    \begin{equation}
        \Dirac^a\check\psi = \langle\diff^a \check\psi, \gamma\circ\tilde\epsilon\rangle\quad \text{(duality pairing along }T\check M\otimes S\check M\text{)},
    \end{equation}
    where $\tilde\epsilon$ acts as the musical isomorphism.
    Hence, $\Dirac^a\check\psi$ is again a section of $S\check M\otimes\ad(\check P)\to \check M$. We plug in the definition of $\check\psi$ such that
    \begin{align}
        \sqrt n \cdot \Dirac^a\check\psi &= \langle\diff^a i^*\langle \psi,\tilde\gamma\rangle,\gamma\circ\tilde\epsilon\rangle\\
        &= i^*\langle\diff^\mathsf a \phi,\tilde\gamma\circ\gamma\circ\tilde\epsilon\rangle,
    \end{align}
    where we used that $\gamma$ and $\tilde\epsilon$ are covariantly constant.
\end{proof}
This allows us to compute the second appearing term arising in the sYM setup:
\begin{lemma}\label{lemma:dirac-laplace}
    Let $\Delta_\epsilon^\mathsf a = \epsilon^{XY} \nabla_X \circ\nabla_Y$, let $a=i^*\mathsf a$, and let $\check\psi = i^*\psi$. It holds
    \begin{equation}
        \big(\Delta^\mathsf a_\epsilon \psi\big)\big|_{\check M} =
        2\Dirac^{a} \lambda.
    \end{equation}
\end{lemma}
\begin{proof}
    As before, we expand $\tilde\gamma$ in a local frame, such that
    \begin{align}
        \sqrt n \cdot \Delta^\mathsf a_\epsilon \psi &= \epsilon^{XY}\iota_X\diff^\mathsf a\iota_Y\diff^\mathsf a\langle\phi,\tilde \gamma\rangle\\
        &=\epsilon^{XY}
        \langle\iota_Z\iota_X
        \diff^\mathsf a\iota_Y
        \diff^\mathsf a \phi,
        \tilde \gamma (Z^\vee)\rangle \\
        &=\epsilon^{XY}
        \langle\iota_X\iota_Z
        L_Y
        \diff^\mathsf a \phi,
        \tilde \gamma (Z^\vee),\rangle\\
        &=\epsilon^{XY}
        \langle\iota_X
        (\iota_{[Z,Y]}-\diff^\mathsf a\iota_Y\iota_Z)
        \diff^\mathsf a \phi,
        \tilde \gamma (Z^\vee)\rangle\\
        &= -\epsilon^{XY}
        \langle
        \iota_{[Z,Y]}\iota_X
        \diff^\mathsf a \phi,
        \tilde \gamma (Z^\vee)\rangle - \epsilon^{XY}
        \langle\iota_X
        \diff^\mathsf a \iota_Y\iota_Z
        \diff^\mathsf a \phi,
        \tilde \gamma (Z^\vee)\rangle
    \end{align}
    Here, we used that $\epsilon^{XY}\iota_X\iota_Y = \epsilon^{XY}\iota_{[X,Y]} = 0$, eliminating many of the otherwise appearing terms.
    Upon restricting to $\check M$, the first summand equates to
    \begin{align}
        \epsilon^{XY}
        \langle
        \iota_{[Z,Y]}\iota_X
        \diff^\mathsf a \phi,
        \tilde \gamma (Z^\vee)\rangle\big|_{\check M} &= 
        \epsilon^{XY}
        \langle\iota_X\diff^\mathsf a \phi,\tilde\gamma(Z^\vee)\otimes\Gamma(Z,Y)\rangle\big|_{\check M}\\
        &=\epsilon^{XY}\langle
        \iota_X\diff^\mathsf a\phi,\tilde\gamma\circ\gamma(Y)\rangle\big|_{\check M}\\
        &=-\langle\diff^\mathsf a\phi,\tilde\gamma\circ\gamma\circ\tilde\epsilon\rangle\big|_{\check M}\\
        &=-\sqrt n \cdot \Dirac^{a}\lambda
    \end{align}
    by the preceding lemma. For the second summand, note that
    \begin{align}
        -\epsilon^{XY}
        \langle\iota_X
        \diff^\mathsf a \iota_Y\iota_Z
        \diff^\mathsf a \phi,
        \tilde \gamma (Z^\vee)\rangle\big|_{\check M}&=
        \epsilon^{XY}
        \langle\iota_X
        \diff^\mathsf a \iota_{[Y,Z]} f,
        \tilde \gamma (Z^\vee)\rangle\big|_{\check M}\\
        &=\epsilon^{XY}
        \langle L_X\iota_{[Y,Z]} f,
        \tilde\gamma(Z^\vee)\rangle\big|_{\check M}\\
        &=\epsilon^{XY}
        \langle \iota_{[X,[Y,Z]]} f + \iota_{[Y,Z]}\iota_X\diff^\mathsf a \phi,
        \tilde\gamma(Z^\vee)\rangle\big|_{\check M}\\
        &=\epsilon^{XY}\langle\iota_{[X,[Y,Z]]}f,\tilde\gamma(Z^\vee)\rangle\big|_{\check M} + \sqrt n \cdot \Dirac^a\lambda
    \end{align}
    due to lemma 3, the Cartan relations and the fact that $\iota_X f = \iota_Z f = 0$ as $X,Z$ are spinorial sections. Note that upon restricting to the ordinary spacetime, $[X,[Y,Z]]$ is purely odd -- however, $i^*f$ is a purely even $2|0$-form, hence, the first term vanishes. Adding up the remaining contributions, we obtain
    \begin{equation}
        \Delta^\mathsf a_\epsilon \psi\big|_{\check M} =
        2\Dirac^{a} \lambda, 
    \end{equation}
    as claimed.
\end{proof}
\subsection*{\texorpdfstring{$d=4$ Frame Computations}{d=4 Frame Computations}}
We now conduct similar computations as before, but suitably adapted for the $d=4$ case. 
\begin{lemma}\label{lemma:d-4-curvature}
Let $\psi = \chi + \bar\chi$ be the decomposition of the sYM spinor into chiral and antichiral components. Let $\varphi^s \in \mathcal A^s_{\hat Q}$ be a spinorial connection, and let $A=\varphi^s\oplus \mathsf a$ be the product connection on the bundle $\hat Q\times_M P \to M$. Denote by $\nabla\coloneqq\nabla^A$ the covariant derivative on $\ad(P)\otimes F^{1,0}$, where $F^{1,0}\to M$ is the chiral distribution, corresponding to $\mathcal F^{1,0}\subset \mathcal T _M$. It holds
\begin{equation}
\sum_{XY}\epsilon^{XY} (\epsilon\otimes\kappa) (\nabla_X \chi, \nabla_Y\chi)|_{\check M} = \|F^a\|^2 ,
\end{equation}
where $a=i^*\mathsf a$ is the restricted gauge connection.
\end{lemma}
\begin{proof}
We begin by computing the covariant derivatives of the spinor fields. Note that in the chiral picture, $F^\mathsf a$ decomposes as
\begin{equation}
F^\mathsf a = f^\mathsf a \oplus \zeta^\mathsf a\oplus \bar\zeta^\mathsf a,
\end{equation}
where $f^\mathsf a$ is a $2|0,0$-form, $\zeta^\mathsf a$ is a $1|1,0$-form, and $\bar\zeta^\mathsf a$ is a $1|0,1$-form. The $\Gamma$ maps restrict to $\wedge^2\mathcal F^{1,0} \to\mathcal B$ such that 
\begin{equation}
\begin{tikzcd}
\zeta\arrow[mapsto]{d}{\frac{1}{\sqrt{n}}\langle\cdot,\tilde\gamma\rangle} &
\bar\zeta\arrow[mapsto]{d}{\frac{1}{\sqrt{n}} \langle\cdot,\tilde\gamma\rangle}\\
\chi & \bar\chi.
\end{tikzcd}
\end{equation}
In what follows, we shall consider a fixed connection $\mathsf a$ and drop it from the notation in the component fields. It holds
\begin{align}
\sqrt{n}\nabla_X \chi &= \iota_X\diff^\mathsf a \langle\zeta,\tilde\gamma\rangle\\
&=\sum_{Y}\langle\iota_Y\iota_X\diff^\mathsf a \zeta,\tilde\gamma(Y^\vee)\rangle\\
&=-\sum_{Y}\langle\iota_Y\iota_X\diff^\mathsf a (f + \bar\zeta),\tilde\gamma(Y^\vee)\rangle\\
&=-\sum_{Y}\langle\iota_{[Y,X]}\diff^\mathsf a (f + \bar\zeta),\tilde\gamma(Y^\vee)\rangle
\end{align}
as $\iota_X\bar\zeta = \iota_X f = 0$.
Upon restricting to $\check M$, the commutator on $\mathcal F^{1,0}$ turns into the $\gamma$ map, such that
\begin{equation}
\sqrt{n}\nabla_X \chi = \iota_X\diff^\mathsf a \langle\zeta,\tilde\gamma\rangle|_{\check M} =
-\langle f + \bar \zeta,\tilde\gamma \circ \gamma(X)\rangle|_{\check M} = 
-\langle f ,\tilde\gamma \circ \gamma(X)\rangle|_{\check M}
\end{equation}
as the duality pairing is along $\wedge^2 \mathcal B |_{\check M}$, and $\bar\zeta$ is a $1|0,1$-form.\par
It follows that
\begin{equation}
\epsilon^{XY} (\epsilon\otimes\kappa) (\nabla_X \chi, \nabla_Y\chi)|_{\check M}
= \epsilon^{XY} (\epsilon\otimes\kappa)
 \big( \langle f,\tilde\gamma \circ \gamma(X)\rangle,
 \langle f,\tilde\gamma \circ \gamma(Y)\rangle \big)\big|_{\check M} = \|F^{a}\|^2,
\end{equation}
where summation over $X,Y$ is implied, and the last equality can be verified analogously to the $d=3$ case.
\end{proof}
Finally, the computation giving rise to the Dirac operator goes as follows.
\begin{lemma}\label{lemma:d-4-dirac}
Let the notation be as in the preceding lemma. It holds
\begin{equation}
\big(\Delta^\mathsf a_{\epsilon} \chi \big)\big|_{\check M} \coloneqq \Big(\epsilon^{XY}\nabla_X\nabla_Y \chi\Big)\Big|_{\check M} = -\Dirac^a\big(\chi|_{\check M}\big).
\end{equation}
\end{lemma}
\begin{proof}
The proof is completely analogous to the $d=3$ case. Nevertheless, we state it for completeness.
    \begin{align}
        \sqrt n \cdot \Delta^\mathsf a _\epsilon \chi &= \epsilon^{XY}\iota_X\diff^\mathsf a\iota_Y\diff^\mathsf a\langle\zeta,\tilde \gamma\rangle\\
        &=\epsilon^{XY}
        \langle\iota_Z\iota_X
        \diff^\mathsf a\iota_Y
        \diff^\mathsf a \zeta,
        \tilde \gamma (Z^\vee)\rangle \\
        &=\epsilon^{XY}
        \langle\iota_X\iota_Z
        L_Y
        \diff^\mathsf a \zeta,
        \tilde \gamma (Z^\vee),\rangle\\
        &=\epsilon^{XY}
        \langle\iota_X
        (\iota_{[Z,Y]}-\diff^\mathsf a \iota_Y\iota_Z)
        \diff^\mathsf a \zeta,
        \tilde \gamma (Z^\vee)\rangle\\
        &= -\epsilon^{XY}
        \langle
        \iota_{[Z,Y]}\iota_X
        \diff^\mathsf a \zeta,
        \tilde \gamma (Z^\vee)\rangle - \epsilon^{XY}
        \langle\iota_X
        \diff^\mathsf a \iota_Y\iota_Z
        \diff^\mathsf a \zeta,
        \tilde \gamma (Z^\vee)\rangle
    \end{align}
    Here, we used that $\epsilon^{XY}\iota_X\iota_Y = \epsilon^{XY}\iota_{[X,Y]} = 0$, eliminating many of the otherwise appearing terms.
    Upon restricting to $\check M$, the first summand equates to
    \begin{align}
        \epsilon^{XY}
        \langle
        \iota_{[Z,Y]}\iota_X
        \diff^\mathsf a \zeta,
        \tilde \gamma (Z^\vee)\rangle\big|_{\check M} &= 
        \epsilon^{XY}
        \langle\iota_X\diff^\mathsf a \zeta,\tilde\gamma(Z^\vee)\otimes\Gamma(Z,Y)\rangle\big|_{\check M}\\
        &=\epsilon^{XY}\langle
        \iota_X\diff^\mathsf a\zeta,\tilde\gamma\circ\gamma(Y)\rangle\big|_{\check M}\\
        &=-\langle\diff^\mathsf a\zeta,\tilde\gamma\circ\gamma\circ\tilde\epsilon\rangle\big|_{\check M}\\
        &=-\sqrt n \cdot \Dirac^{a}(\chi|_{\check M})
    \end{align}
    by the preceding lemma. For the second summand, note that
    \begin{align}
        -\epsilon^{XY}
        \langle\iota_X
        \diff^\mathsf a \iota_Y\iota_Z
        \diff^\mathsf a \zeta,
        \tilde \gamma (Z^\vee)\rangle\big|_{\check M}&=
        \epsilon^{XY}
        \langle\iota_X
        \diff^\mathsf a \iota_{[Y,Z]} (f+\bar\zeta),
        \tilde \gamma (Z^\vee)\rangle\big|_{\check M}\\
        &=\epsilon^{XY}
        \langle L_X\iota_{[Y,Z]} (f+\bar\zeta),
        \tilde\gamma(Z^\vee)\rangle\big|_{\check M}\\
        &=\epsilon^{XY}
        \langle \iota_{[X,[Y,Z]]} (f+\bar\zeta) + \iota_{[Y,Z]}\iota_X\diff^\mathsf a \zeta,
        \tilde\gamma(Z^\vee)\rangle\big|_{\check M}\\
        &=\epsilon^{XY}\langle\iota_{[X,[Y,Z]]}(f+\bar\zeta),\tilde\gamma(Z^\vee)\rangle\big|_{\check M} + \sqrt n \cdot \Dirac^a(\chi|_{\check M})
    \end{align}
    due to lemma 3, the Cartan relations and the fact that $\iota_X f = \iota_X \bar\zeta = 0$ as $\bar\zeta$ is a $1|0,1$ form, and upon restricting to the ordinary spacetime, $[X,[Y,Z]]$ is supported on $\mathcal F^{1,0}|_{\check M}$. Adding up the remaining contributions, we obtain
    \begin{equation}
        \Delta^\mathsf a_\epsilon \chi\big|_{\check M} =
        2\Dirac^{a} (\chi|_{\check M}), 
    \end{equation}
    as claimed.
\end{proof}

\end{document}